\documentclass[final,5p,times,twocolumn]{elsarticle}
\journal{Signal Processing}

\usepackage[cmex10]{amsmath}

\usepackage{array}
\usepackage{arydshln}

\usepackage{eqparbox}

\usepackage[caption=false,font=footnotesize]{subfig}

\usepackage{amsfonts}
\usepackage{amssymb}
\usepackage{rotating}
\usepackage{balance}
\usepackage{enumerate}
\usepackage{latexsym}
\usepackage{mathtools}
\usepackage{url}
\usepackage{floatrow}

\usepackage[cmex10]{amsmath}
\usepackage{mathrsfs}
\usepackage{multirow}

\usepackage{leftidx}
\usepackage{stmaryrd}
\usepackage{dsfont}

\usepackage{algorithm}
\usepackage{algorithmic}

\usepackage{esint}

\usepackage{hhline}
\newcommand{\tens}[1]{\boldsymbol{\mathcal{#1}}}

\newcommand{\matr}[1]{\boldsymbol{#1}}

\newcommand{\vect}[1]{\boldsymbol{#1}}

\makeatletter
\renewcommand{\vec}{\mathop{\operator@font vec}}   %
\newcommand{\vecT}{\mathop{\operator@font vecT}}   %
\newcommand{\Unvec}{\mathop{\operator@font Unvec}}   %
\newcommand{\Span}{\mathop{\operator@font span}}
\makeatother

\newcommand{\ZT}{\mathcal{Z}}
\newcommand{\IZT}{\mathcal{Z}^{-1}}
\newcommand{\FT}{\mathcal{F}}
\newcommand{\IFT}{\mathcal{F}^{-1}}
\newcommand{\HT}{\mathcal{H}}
\newcommand{\IHT}{\mathcal{H}^{-1}}

\newcommand{\cep}[1]{\hat{#1}}
\newcommand{\hel}[1]{{#1}^{(1)}}

\usepackage{amsthm}
\theoremstyle{plain}
\theoremstyle{definition}
\newtheorem{defi}{Definition} 
\newtheorem{theo}{Theorem} 
\newtheorem{propo}{Proposition} 
\newtheorem{corol}{Corollary}

\usepackage{color}
\definecolor{darkred}{rgb}{0.7,0,0}
\definecolor{darkgreen}{rgb}{0,0.46,0}

    \makeatletter
    \def\ps@pprintTitle{%
      \let\@oddhead\@empty
      \let\@evenhead\@empty
      \let\@oddfoot\@empty
      \let\@evenfoot\@oddfoot
    }
    \makeatother

\begin{document}\sloppy
\begin{frontmatter}
\title{Multidimensional factorization \\ through helical mapping}
%\subtitle{Application to helioseismology}

\author{Francesca Raimondi\fnref{gipsa}}
\ead{francesca.raimondi@gipsa-lab.grenoble-inp.fr}
\author{Pierre Comon\fnref{gipsa}}
\author{Olivier Michel\fnref{gipsa}}
\author{Umberto Spagnolini\fnref{polimi}}

\fntext[gipsa]{GIPSA-Lab (Department of Image and Signal-processing), CNRS, Universit\'e Grenoble Alpes, 38400 Saint Martin d'H\`eres, France}
\fntext[polimi]{Dipartimento di Elettronica Informazione e Bioingegneria, Politecnico di Milano, I-20133 Milano, Italy}

%This work was supported by the ERC Grant AdG-2013-320594 ``DECODA''. 
%

\begin{abstract}
This paper proposes a new perspective on the problem of multidimensional spectral factorization, through helical mapping: $d$-dimensional ($d$D) data arrays are vectorized, processed by $1$D cepstral analysis and then remapped onto the original space. Partial differential equations (PDEs) are the basic framework to describe the evolution of physical phenomena. We observe that the minimum phase helical solution asymptotically converges to the $d$D semi-causal solution, and allows to decouple the two solutions arising from PDEs describing physical systems. We prove this equivalence in the theoretical framework of cepstral analysis, and we also illustrate the validity of helical factorization through a $2$D wave propagation example and a $3$D application to helioseismology.
\end{abstract}


%\ifCLASSOPTIONpeerreview
%\begin{center} \bfseries EDICS Category: SAM-TNSR, SSP-SSEP, NET-FUSE  \end{center}
%\fi
%\IEEEpeerreviewmaketitle
\end{frontmatter}
%%%%%%%%%%%%%%%%%%%%%%%%%%%%%%%%%%%%%%%%%%%%%%%%%%%%%%%%%%%

\section{Introduction}

Wavefield processing has been applied to several fields of physical sciences. 
Inverse problems include the estimation of the impulse response of a physical system, such as the earth response to an ideal impulse-like seismic source. 
On the other hand, seismic migration consists in inferring the signal that would be measured at any depth, starting from data recorded on the surface \cite{Gari79:ieeetsp}. More generally, multidimensional digital filters are extensively used in remote sensing, image processing, medical imaging and geophysics. In array processing, multidimensional filters have been used to separate seismic waves based on their polarization or propagation velocity differences \cite{DonnNS08:ieeetsp}.

Physical filters are mostly ruled by partial differential equations (PDEs) that can be studied, in some cases, as linear operators through Fourier Transform. 
Wavefield propagation through an homogeneous medium can be represented by a linear PDE with constant coefficients, and hence the inverse problem is reduced to a $d$D deconvolution (and then, to multidimensional linear filtering).  If we assume a minimum phase impulse response, and if the source is white in all its dimensions, blind deconvolution is equivalent to spectral factorization \cite{RickC00:SP}, which can be tackled through homomorphic deconvolution \cite{TaxtF99:ieeetuffc, taxtS01:ieeetuffc}. This approach has been extended to parametric autoregressive processes through linear prediction (\textit{predictive deconvolution}) \cite{TakaNFDSLRT12:ieeespm}, and cepstral analysis \cite{SarpC02:ieeeicdsp, Kizi07:ieeetcs}. 

In particular, spectral factorization consists in separating causal and anti-causal components in physical systems through decomposition of spectral density functions as the product of minimum phase and maximum phase terms \cite{DudgM84:PHSPS}. 
Cepstrum analysis eases the design of causal filters \cite{OppeSS68:ieee, OppeS04:ieeespm}. The latter have been extended to two dimensions \cite{DudgM84:PHSPS,EkstW76:ieeetassp, ChanA78:ieeetcs, YangL07:ieeetsp, PeiL11:ieeetcs}, and generalized to the multidimensional case \cite{GoodE80:ieeetac}, through the definition of $d$D \textit{semi-causality}.
When dealing with  $d$D spectral factorization with $d\geq2$, problems arise from the lack of a unique definition of  $d$D minimum phase, leading to the existence of multiple possible solutions \cite{EkstW76:ieeetassp}.
 Mersereau and Dudgeon \cite{MersD74:ieeetsp} propose an alternative approach to describe $2$D signals, based on a transformation of $2$D sequences into $1$D vectors, such that linear convolution becomes \emph{helical} (cf. Section \ref{sec-helix}). 
Helical coordinates were applied later in \cite{RickC00:SP} to blind deconvolution in helioseismology, through spectral factorization. 

The purpose of the present paper is to prove that helical mapping and spectral factorization are asymptotically equivalent.  We will show that the $1$D causal solution after helical mapping is not only recursively computable and stable, but asymptotically convergent to the \textit{semi-causal} $d$D solution, after inverse remapping. 

Sections \ref{sec-helix} and \ref{sec-spectr-fact} introduce helical mapping, and spectral factorization, respectively. Section \ref{sec-proof} proves the asymptotic equivalence of helical spectral factorization with its $d$D counterpart, whereas Section \ref{sec-back-prop} shows how it applies to the case of wavefield propagation. Finally, Section \ref{sec-application} presents an example of causal physical filters (the anti-causal component representing the reversed time solution in propagating systems), and an application to helioseismology.

\section{The helical coordinate system} \label{sec-helix}
Physical fields are generally sampled through space (in a domain $\Omega \subset \mathbb{R}^3$) and time (in a domain $U \subset \mathbb{R}^+$), resulting in a $d$D data cube, where $d\leq4$. Under certain conditions (translational invariance through homogeneous and linear media), the underlying processes of propagation can be modeled   as Linear Shift Invariant (LSI) filters (including Linear Time Invariant (LTI) and space invariant systems). Moreover, thanks to multilinearity, the measured data cube can be represented as a tensor. 

A tensor of order $d$ is an element of the outer product of vector spaces $\mathbb{S}_1 \otimes \cdots \otimes \mathbb{S}_d$, and can be represented by a multi-way array (or multidimensional matrix), once bases of spaces $\mathbb{S}_i$ have been fixed. The order $d$ of a tensor corresponds to the number of dimensions of the physical system. The \textit{mode-$i$ fiber} of a tensor is a vector obtained when all indices are fixed except the $i$-th. It is often useful to represent a tensor in  matrix form \cite{Como14:ieeespm}: the \textit{mode-$i$ unfolding} of a tensor $\tens{Y} \in \mathbb{R}^{I_1 \times I_2 \times \cdots \times I_d}$ reorders its elements, through arranging the mode-$i$ fibers into columns of a matrix denoted $\matr{Y}_{\!(i)}$. 
Furthermore, it is convenient to represent tensors as vectors: the vectorization of a tensor $\tens{Y}$ is generally defined as a vectorization of the associated mode-1 unfolding matrix, that stacks the columns of $\matr{Y}_{\!(1)}$ into a vector $\vect{y} \in \mathbb{R}^{I_1 I_2 \cdots I_d}$. Unfolding and then vectorizing a tensor are equivalent to gradually reducing its order: for instance, a $3$D data cube is at first transformed into a matrix and then into a vector. 
It is clear that there exist potentially multiple ways of unfolding and vectorizing tensors, thus reducing their order, all related to the definition of a particular ordering relation.

Since causality of a LTI filter is related to the implicit order of the computation in convolution, $d$D causality is associated with the existence of an ordering relation organizing the elements of the multidimensional data cube.
For $1$D systems, there is only the natural (or reversed) ordering (\textit{i.e.} \textit{fully ordered} computation of a linear transformation such as convolution). For $d$D systems, the computation is only \textit{partially ordered}, as there are multiple possible ordering relations \cite{DudgM84:PHSPS}. 
In order to implement any $d$D linear transformation ($d$D convolution, $d$D filtering, etc.), there is a need to define an ordering map $p = I (n_1,n_2, \dots, n_d)$. Thus, if $p' = I (n_1',n_2', \dots, n_d')$, $p<p'$ implies that the output at $(n_1,n_2, \dots, n_d)$ will be computed before the output at $(n_1',n_2', \dots, n_d')$.

One of the simplest ordering relations is the helical transformation of a tensor $\tens{Y}$, that stacks all the elements of any mode-$i$ unfolding, either row-wise or column-wise. Thus, the \textit{helix} is a form of vectorization. 
Therefore, there exist several possible helical transforms of a tensor, corresponding to a progressive reduction of the order.  
For instance, the helical transform of a $2$D sequence $f(m,n), \, m \in \mathbb{N}, \, 0 \leq n \leq N-1$ can be represented as a  row-wise  invertible mapping: 
%\begin{small}\label{eq-helix-1}
\begin{equation}\nonumber
\phi_{2} :  \mathbb{N} \times [0:N-1] \longrightarrow \mathbb{N}, \: (m,n) \longmapsto p = Nm+n
\end{equation}
%\end{small}
which corresponds to concatenating  the rows  of the matrix $f(m,n)$, with condition of invertibility $N < \infty$. Alternatively,  the column-wise invertible mapping of a sequence $f(m,n), \, 0 \leq m \leq M-1, \, n \in \mathbb{N}$ can be written as
%\begin{small}
\begin{equation}\label{eq-helix-2}
\phi_{2} :  [0:M-1]  \times \mathbb{N} \longrightarrow \mathbb{N}, \: (m,n) \longmapsto p = m+Mn
\end{equation}
%\end{small}
which corresponds to concatenating  the columns  of  matrix $f(m,n)$, with condition of invertibility $M < \infty$.
Equivalently, one $3$D helical mapping of a data cube 
$f(l,m,n), \text{ for } 0 \leq l \leq L-1, 0 \leq m \leq M-1, n \in \mathbb{N}$ is given by
%\begin{small}
\begin{equation}\nonumber
\begin{aligned}
\phi_{3}: & [0:L-1] \times [0:M-1] \times \mathbb{N} \longrightarrow \mathbb{N} \\
& (l,m,n) \longmapsto p = L(M n + m) + l
\end{aligned}
\end{equation}
%\end{small}
with condition of invertibility $L,M < \infty$. 
In the seventies, Mersereau and Dudgeon \cite{MersD74:ieeetsp}  defined a helical convolution that transforms $2$D convolution through helical periodicities, showing that helical convolution is numerically equal to its $2$D counterpart. 

\section{Cepstral factorization} \label{sec-spectr-fact}
\paragraph{One-dimensional case} 
A $1$D sequence $s(n)$ is \textit{causal} if $s(n) = 0, \text{ for } n<0$, and \textit{minimum phase} if all the poles and zeros of the Z-transform $S(z) = \ZT\{s(n)\}$ are inside the unit circle $\{\lvert z \lvert < 1\}$.
If a sequence is minimum phase, it is also causal. Moreover, minimum phase sequences are also \textit{minimum-delay} (all their energy is concentrated close to time origin $n=0$), they are absolutely summable, and their inverses are both causal and absolutely summable \cite{DudgM84:PHSPS}. 

A means to investigate the question of causality and minimum phase  in relation to  spectral factorization is homomorphic analysis. The homomorphic transform \begin{small}$\HT= \IZT \! \circ \log \circ \ZT$\end{small} (\textit{i.e.} the inverse Z-transform of the complex logarithm of the Z-transform) with inverse \begin{small}$\IHT = \IZT \! \circ \exp \circ \ZT$\end{small} has the advantage of converting convolutions into sums:  \begin{small}$\HT\{s_1 \ast s_2\} = \HT\{s_1\} + \HT\{s_2\}$\end{small}.

The stability condition for a system with impulse response $s$ is that its transfer function $S$ converges on a region containing the unit circle $\{z=e^{i \omega}\}$, or, equivalently, its domain of convergence includes the locus $\{ \lvert z \lvert =1 \}$. 
In this case, the FT, $\FT\{\cdot\}$, of a sequence can be defined as the restriction of its Z-transform on the unit circle. $\HT$ is then defined as the Inverse Fourier Transform (IFT) of the complex logarithm of its Fourier Transform (FT), \begin{small}$\HT = \IFT \! \circ \log \circ \FT$\end{small}, after phase unwrapping of the complex logarithm \cite{Dudg77:ieeetassp}. The complex cepstrum of a limited sequence $s(n), 0 \leq n \leq N$, can be calculated through the Discrete FT (DFT), as an aliased version of the true cepstrum \cite{OppeSS68:ieee}.

If $s(n)$ is the autocorrelation of a sequence $x(n)$ assumed stationary, its homomorphic transform $\cep{s} = \HT\{s\}$ is called \textit{complex cepstrum} and corresponds to the IFT of the logarithm of the spectrum,  $\cep{s} = \IFT \{ \log(\lvert X(\omega) \lvert^2)\}$.  We will refer in what follows to positive definite or autocorrelation sequences $s(n)$. 
The complex cepstrum is useful to characterize causality: a sequence $s(n)$ is minimum phase if its cepstrum is causal \cite{SteiD77:ieeeicassp, PedeAD10:ieeeicassp}: $\cep{s}(n)=0$ for $n<0$. 

Inversely, maximum phase sequences can be defined as minimum phase sequences reversed in time\footnote{
A \textit{maximum phase} sequence is anticausal with an anticausal inverse and anticausal complex cepstrum: $\cep{s}(n)=0$ for $n>0$.}, and any absolutely summable signal, if conveniently shifted in time, can be expressed as the convolution between minimum and maximum phase parts \citep{DudgM84:PHSPS}.  As a result, its complex cepstrum is the sum of causal and anti-causal parts, and it is absolutely summable: $s(n) = s_+(n) \ast s_-(n)$
corresponds to a product in the frequency domain $S(\omega) = S_+(\omega) S_-(\omega)$, and to a sum in the cepstrum domain $\cep{s}(n) = \cep{s}_+(n) + \cep{s}_-(n)$.  In particular, the poles $z_i$ of $S(z)$ such that $\lvert z_i \lvert < 1$ are associated with the causal part of the cepstrum, whereas the poles such that $\lvert z_i \lvert > 1$ correspond to the anti-causal part of the cepstrum \cite{SteiD77:ieeeicassp}. 
The $1$D factorization problem consists in decomposing a real (or zero-phase) sequence (such as a power spectral density) into minimum and maximum phase terms. 

\paragraph{Higher-dimensional case} 
The concept of minimum phase solutions of the spectral factorization problem was extended to $2$D signals in \cite{EkstW76:ieeetassp}. However, the derivation of the concepts of $2$D causality and minimum phase from the $1$D equivalent is not straightforward, due to the existence of several ordering relations in the $(x,y)$ plane.  
From the $2$D Z-transform  
%\begin{small}
\begin{equation}\label{eq-2D-z-transform}
S(z_1, z_2) = \sum_{m=-\infty}^{\infty} \sum_{n=-\infty}^{\infty} s(n_1,n_2) z_1^{-n_1} z_2^{-n_2},
\end{equation}
%\end{small}
the stability condition for a system with impulse response $s$ is that its transfer function $S$ converges on a region containing the unit bicircle ($z_1=e^{i \omega_1}, z_2=e^{i \omega_2}$), or, equivalently, its domain of convergence includes the locus $\{\lvert z_1 \lvert =1, \lvert z_2 \lvert =1\}$.
Thus, the $2$D FT, $\FT\{\cdot\}$, of a $2$D sequence is defined as the restriction of its Z-transform on the unit bicircle.
In what follows, we shall consider spectral density functions $S(z_1, z_2)$ as the $2$D Z-transform of autocorrelations $s(n_1,n_2)$.
The \textit{$2$D spectral factorization} is a decomposition of the $2$D Z-transform $S(z_1, z_2)$ into factors that are free of poles and zeros in certain regions of $\mathbb{C}^2$.
In particular, a sequence $s(n_1,n_2)$ is said to be $\textit{min-min}$ phase, if none of the poles and zeros of $S(z_1, z_2)$ lie in the closed domain $\{|z_1| \geq 1, |z_2| \geq 1\}$; \textit{min-mix} phase if none of its poles or zeros lie in $\{|z_1| \geq 1, |z_2| = 1\}$. See \cite{EkstW76:ieeetassp} for further details. 
We hereby refer to the min-min phase as a \textit{strict} $2$D minimum phase. 

Analogously, in the 3D case, a sequence $s(n_1, n_2, n_3)$ is defined as \textit{min-min-min} phase if poles and zeros of its Z-transform $F(z_1, z_2, z_3)$ do not lie in $\{|z_1| \geq 1, |z_2| \geq 1,  |z_3| \geq 1\}$.

$2$D causality can be studied through the $2$D complex cepstrum  \cite{DudgM84:PHSPS}, which is defined through $2$D homomorphic transform $\HT = \IZT \circ \log \circ \ZT$:
\begin{small}
\begin{equation} \nonumber
\cep{s}(n_1,n_2) = -\frac{1}{4 \pi^2}\ointctrclockwise_{\lvert z_1 \lvert =1}\ointctrclockwise_{\lvert z_2 \lvert =1} \log [S(z_1,z_2)] \cdot z_1^{n_1-1}z_2^{n_2-1} dz_1 dz_2
\end{equation}
\end{small}
Then,  provided that the unit bicircle is confined in the definition domain of the $2D$ Z-transform $S(z_1,z_2)$, that $\lvert S(e^{i \omega_1},e^{i \omega_2}) \lvert \neq 0 \text{ for } -\pi \leq \omega_1, \omega_2 \leq \pi$ and that the phase of the signal has been adjusted to be continuous and periodic in both frequency variables $\omega_1$ and $\omega_2$ (\textit{i.e.} $2$D phase unwrapping), we can write\footnote{The complex cepstrum of a time limited sequence $s(n_1,n_2), 0 \leq n_1 \leq N_1, 0 \leq n_2 \leq N_2$ can then be calculated through the $2$D Discrete Fourier Transform (DFT) \cite{Dudg77:ieeetassp}, as a spatially aliased version of the cepstrum \cite{DudgM84:PHSPS}.}
%\begin{small}
\begin{equation} \nonumber
\cep{s}(n_1,n_2) = \frac{1}{4 \pi^2}\!\int_{-\pi}^{\pi}\int_{-\pi}^{\pi} \log [S(e^{i \omega_1},e^{i \omega_2})] \cdot e^{i \omega_1 n_1 + i \omega_2 n_2} d\omega_1 d\omega_2  
\end{equation}
%\end{small}

Based on the definition of \textit{non-symmetric half plane (NSHP)} as a region of the form $\{n_1 \geq 0, n_2 \geq 0\} \cup \{n_1 > 0, n_2 < 0\}$ or $\{n_1 \geq 0, n_2 \leq 0\} \cup \{n_1 > 0, n_2 > 0\}$ or their rotations, an \textit{admissible region} is the Cartesian product of a sector
\footnote{A \textit{sector} $S(\alpha, \beta)$ is defined in polar form as $S(\alpha, \beta) = \{(r,\theta) \lvert r>0, \alpha < \theta < \beta\}$.}
and a NSHP. 
Before introducing the subject of multidimensional spectral factorization, we must restate some preliminary results from \cite{EkstW76:ieeetassp}.

\begin{defi}
Given a sequence $x(n_1,n_2)$, a \textit{projector operator} $P$ is defined as the multiplication by a window $w_P(n_1,n_2)$ with support $\mathcal{R}_w \subset \mathbb{R}^2$.
\end{defi}

\begin{propo}\label{prop-projection}
Let $s(n_1, n_2)$ be an autocorrelation, or a non negative definite sequence, and its Z-transform be the spectral function $S(z_1, z_2)$. The \textit{$2$D spectral factorization} of $s$ results in a decomposition of the range of its cepstrum $\cep{s}$ into admissible regions, through a set of projections operators $P_k$ whose sum is the identity ($\prod_k{w_k} = 0, \sum_k{w_k} = 1$). 
\end{propo}
 Prop. \ref{prop-projection} is based on a result stated by the theorem  below, whose proof can be found in \cite{EkstW76:ieeetassp}: 
\begin{theo} \label{theo-l1}
Let $\cep{s}$ be the cepstrum of a sequence $s$ (assuming $s$ is absolutely summable), and let $P(\cep{s})$ be its projection onto an admissible region, then $s_{P} = \IHT\{P(\cep{s})\}$ is recursively computable and stable.
\end{theo}

In particular, a sequence $s(n_1, n_2)$ is min-min phase if its cepstrum is causal, \textit{i.e.} with support $\mathcal{S}_{\cep{s}}$ included in the first quadrant: $\mathcal{S}_{\cep{s}} \subset \mathcal{R}_{++}$, with $\mathcal{R}_{++} \coloneqq \{n_1 \geq 0, n_2 \geq 0\}$; and \textit{semi-minimum phase} if its cepstrum is semi-causal, \textit{i.e.} with support included in the upper NSHP: $\mathcal{S}_{\cep{s}} \subset \mathcal{R}_{\oplus+}$, with $\mathcal{R}_{\oplus+} \coloneqq \{n_1 \geq 0, n_2 \geq 0\} \cup \{n_1 < 0, n_2 > 0\}$.  In the latter case, $s(n_1, n_2)$ is minimum-phase only with respect to the variable $n_2$, as depicted in Figure \ref{fig-2D-causality}. 

Recursive computability is equivalent to the existence of an ordering relation. If the admissible regions coincide with the 4 quadrants, the four projections of the cepstrum onto $\mathbb{R}_{++}, \mathbb{R}_{+-}, \mathbb{R}_{-+}, \mathbb{R}_{--}$ give a four factor decomposition and involve a strong definition of $2$D causality (cf. Figure \ref{fig-quarter-R++}). If the admissible regions coincide with the upper and the lower NSHPs, the two projections of the cepstrum onto $\mathbb{R}_{\oplus+}, \mathbb{R}_{\ominus-}$ yield a two factor decomposition and involve a weaker definition of $2$D semi-causality (cf. Figure \ref{fig-non-sym-HP-R++}) \cite{EkstW76:ieeetassp}\footnote{
with $\mathcal{R}_{+-} \coloneqq \{n_1 \geq 0, n_2 \leq 0\}$, $\mathcal{R}_{-+} \coloneqq \{n_1 \leq 0, n_2 \geq 0\}$, $\mathcal{R}_{--} \coloneqq \{n_1 \leq 0, n_2 \leq 0\}$, and $\mathcal{R}_{\ominus-} \coloneqq \{n_1 \leq 0, n_2 \leq 0\} \cup \{n_1 > 0, n_2 < 0\}$.}. Through the projection onto $\mathcal{R}_{\oplus +}$ and $\mathcal{R}_{\ominus -}$, the cepstrum of the autocorrelation is decomposed into $\cep{s} = \cep{s}_{\oplus+}  + \cep{s}_{\ominus-}$
corresponding, after inverse homomorphic transform, to
$s= \mathcal{H}^{-1} (\cep{s}) = s_{\oplus+}  \ast s_{\ominus-}$.
The two-factor decomposition, based on the definition of NSHPs, is less restrictive than the four factor decomposition, as it can describe the general class of positive definite magnitude functions.
A magnitude function, such as the power spectral density in the spectral factorization problem, can be expressed by a limited number of factors, omitting those with conjugate symmetry.
Then, for the two factor decomposition we have
\begin{small}
$
\lvert s(m,n) \lvert^2 = \lvert s_{\oplus+}(m,n)\lvert^2
$
\end{small}.

\begin{figure}[h!]
\centering
	\subfloat[$\mathcal{R}_{++}$ - $2$D causality]{\includegraphics[scale=0.37]{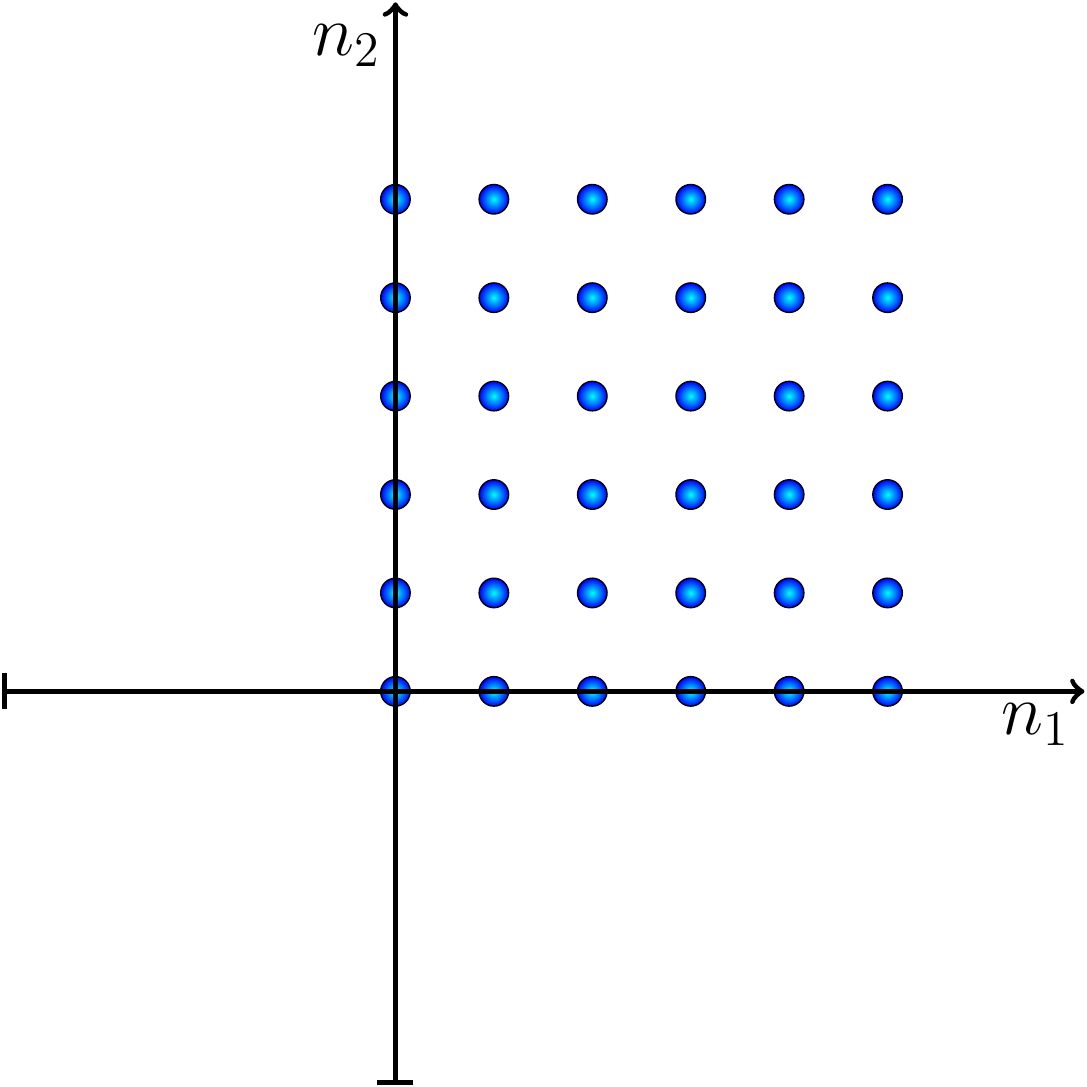}\label{fig-quarter-R++}}\\
	%\hspace{2em}
	\subfloat[$\mathcal{R}_{\oplus+}$ - $2$D semi-causality]{\includegraphics[scale=0.37]{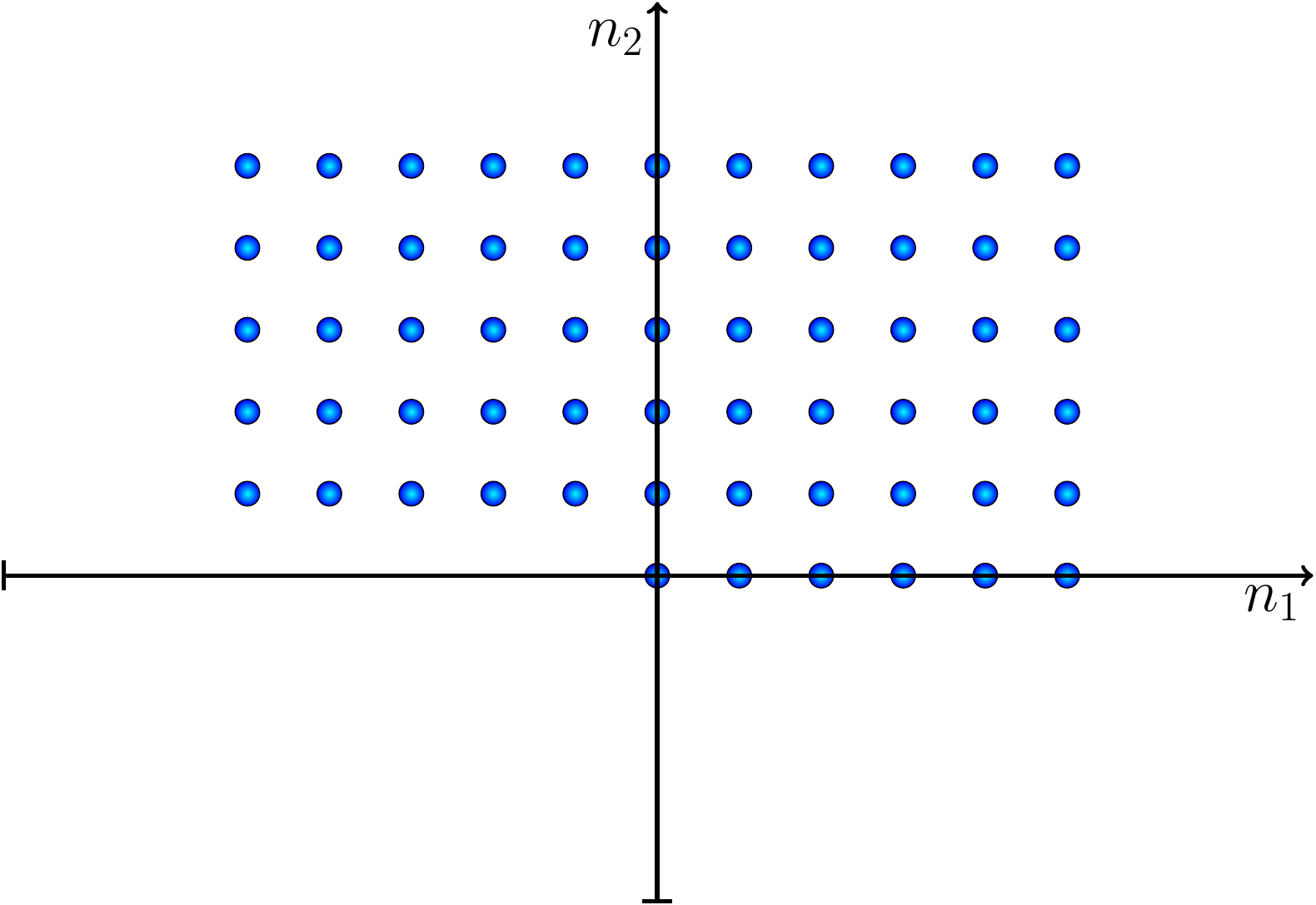}	\label{fig-non-sym-HP-R++}}
\caption{Examples of admissible regions related to $2$D causality}\label{fig-2D-causality}
\end{figure}

Spectral factorization was extended in \cite{GoodE80:ieeetac} to multiple dimensions so as to process data cubes. It is based on multidimensional homomorphic transform, and on the definition of $d$D non-symmetric half-spaces (NSHS), such as the $3$D upper NSHS $\mathcal{R}_{\oplus \oplus +}$. Thus, all the results presented in this section are easily generalized to the $d$D case.

\section{The effect of the helical transform on the multidimensional factorization problem}\label{sec-proof}
 
This Section investigates the effects of the helical ordering relation onto the multidimensional homomorphic analysis.  
We can initially state the following fact, which can be easily generalized to $d$D systems: 
\begin{propo}\label{prop-asymp-equiv}
Let $f(m,n)$ define an absolutely summable $2$D sequence, from which we want to extract the $2$D semi-minimum phase component. Let $\hel{f}(p)$ be the helical transform of $f(m,n)$, after column-wise mapping $p = m +Mn$, and $\hel{f}_+(p)$ its $1$D minimum phase projection, corresponding to causal cepstrum $\hel{\cep{f}}_+(p)$. Then, after inverse helical mapping of $\hel{f}_+(p)$, the solution $f^{hel}_+(m,n)$ is recursively computable and stable, and it tends to be, for large $M$, the semi-minimum phase solution corresponding to semi-causal cepstrum $\cep{f}_+(m,n)$ described in Section \ref{sec-spectr-fact}. 
\end{propo}

\begin{proof}[Proof]
If we consider the discrete variable $m$ bounded by $M<\infty$ and we allow the variable $n$ to be unbounded ($n \in \mathbb{N}$), the helical transformation of the dataset $f(m,n)$, $\hel{f}(p)$, is equivalent to a periodization  of $f(m,n)$  with respect to the  bounded  variable $m$.  
After helical transform, the causal component of the $1$D cepstrum $\hel{\cep{f}}(p)$ is given by the contribution for positive $p$. Through the projection operator in Prop. \ref{prop-projection}, the $1$D complex cepstrum $\hel{\cep{f}} = \hel{\cep{f}}_+ + \hel{\cep{f}}_-$ is decomposed into its causal and anti-causal components, so that $\hel{f} = \hel{f}_+ \ast \hel{f}_-$ and 
$
\hel{\cep{f}}_+(p) \neq 0 \text{ for } p \geq 0
$.
Now, $p \geq 0$ is equivalent to $m+Mn \geq 0$ after helical transform, and then to the NSHP $n>-m/M$ on the $2$D plane $(m,n)$:
$$
\cep{f}^{hel}_+(m,n) \neq 0 \text{ for } n>-m/M
$$
Thus, the helical transformation fixes one particular instance among all the possible canonical factorizations. This means that after inverse mapping of the helical minimum phase solution, the support of $2$D cepstrum becomes an upper NSHP rotated of an angle $\theta = \arctan(-1/M)$. Since any rotated NSHP is an admissible region, according to Theorem \ref{theo-l1}, the resulting $2$D filter $f^{hel}_+(m,n) = \IHT \{\cep{f}^{hel}_+(m,n)\}$ is recursively computable and stable. 
If $M \to \infty$, the rotation becomes irrelevant (as $\theta \to 0$), and the support of the solution $f^{hel}_+(m,n)$ and of its cepstrum coincides with the upper NSHP $\mathcal{R}_{\oplus+}$ defined in Section \ref{sec-spectr-fact} (cf. Figure \ref{fig-inverse-helix-cepstrum}). 
\end{proof}

\begin{small}
\begin{figure}[h!]\centering
	\includegraphics[scale=0.3]{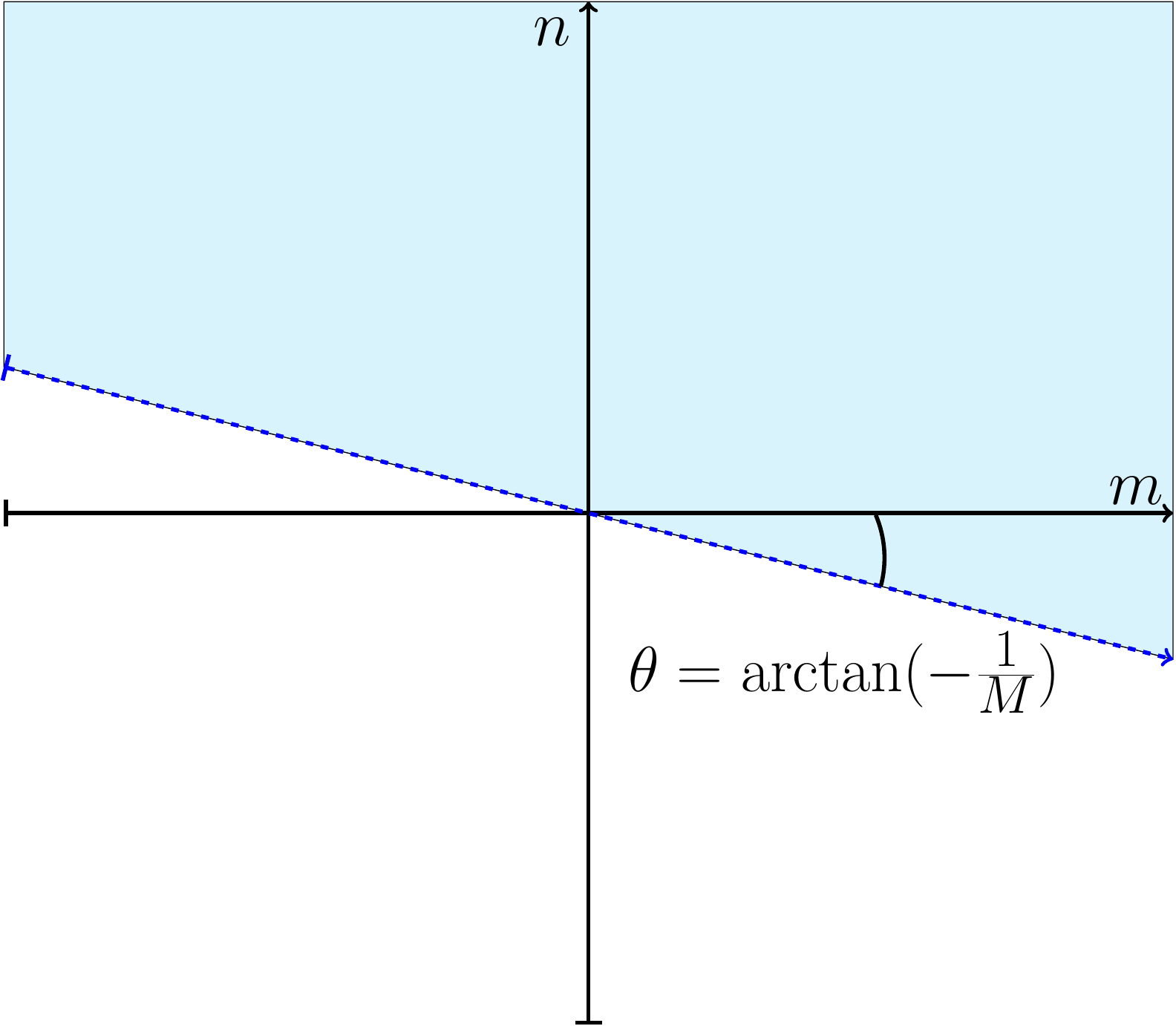}
	\caption{Semi-causal cepstrum after inverse helix transform}\label{fig-inverse-helix-cepstrum}
\end{figure}
\end{small}
Moreover, we can state the following Corollary:
\begin{corol}\label{corol-MP-time}
Since the two factor decomposition of \cite{EkstW76:ieeetassp} leads to a semi-minimum phase term which is minimum phase only in one variable, $M \to \infty$ implies that the helical solution $f^{hel}_+(m,n)$ is minimum phase with respect to the variable $n$.
\end{corol}

\bigskip
Separable functions have noteworthy properties. We can state the following Proposition for $2$D functions (we choose to describe the $2$D case for sake of simplicity, without restricting the generality): 

\begin{propo}
If the $2$D function $f(m,n)$ is separable, \emph{i.e.} if $f(m,n) = u(m) v(n)$, the following two facts hold:
\begin{enumerate}
\item  The $2$D cepstrum of a separable function is given by $\cep{f}(m,n)  = \cep{u} (m) \:\delta(n) + \cep{v} (n) \: \delta(m)$, that is non zero only on the axes of the $(m,n)$ plane. Therefore, $2$D semi-causality of the cepstrum $\cep{f}$ is equivalent to strict $2$D causality (\textit{i.e.} the half-plane support reduces to two lines in the plane $(m,n)$: cf. Figure \ref{fig-inverse-helix-cepstrum-separable}. 
\item The $1$D cepstrum of the vectorized data $\hel{f}(p)$ is given by $\hel{\cep{f}}(p) = \cep{u}(p) + \frac{1}{M} \: \cep{v}\left(\frac{p}{M}\right)$. Helical cepstrum $\hel{\cep{f}}$ is then causal if and only if $1$D cepstra $\cep{u}$ and $\cep{v}$ are both causal, and thus, from $1)$ if and only if $2$D cepstrum $\cep{f}$ is strictly causal. Therefore, the equivalence between strict $2$D minimum phase of $f$ and $1$D minimum phase of its helix $\hel{f}$ is always verified, not only asymptotically.
\end{enumerate}

\begin{proof}[Proof]
Let us define a function $f(m,n), \, 0 \leq m \leq M-1, \, n \in \mathbb{N}$, that is separable with respect to its two variables:
$$
f(m,n) = u(m) \: v(n) 
$$
Then, its Z-transform is also separable in the frequency domain: 
\begin{small}
\begin{equation}\nonumber
\begin{aligned}
F(z_1, z_2) &= \sum_{m=0}^{M-1} \sum_{n=0}^{\infty} f(m,n) z_1^{-m} z_2^{-n} = \\
&= \sum_{m=0}^{M-1} \sum_{n=0}^{\infty} u(m) v(n) z_1^{-m} z_2^{-n} = \\
&= U(z_1) \: V(z_2)
\end{aligned}
\end{equation}
\end{small}
Since $\log [F(z_1, z_2)] = \log U(z_1) + \log V(z_2)$, 
the cepstrum becomes
\begin{small}
\begin{equation}\nonumber
\begin{aligned}
&\cep{f}(m,n) = -\frac{1}{4 \pi^2}\ointctrclockwise_{\lvert z_1 \lvert =1} \ointctrclockwise_{\lvert z_2 \lvert =1} \log [F(z_1, z_2)] z_1^{n_1-1}z_2^{n_2-1} dz_1 dz_2 = \\
&= \frac{1}{4 \pi^2}\int_{-\pi}^{\pi} \int_{-\pi}^{\pi} \log [F(e^{i \omega_1},e^{i \omega_2})] e^{i \omega_1 m}e^{i \omega_2 n} d\omega_1 d\omega_2 = \\
&= \frac{1}{4 \pi^2}\int_{-\pi}^{\pi} \int_{-\pi}^{\pi} \left[\log U(e^{i \omega_1}) + \log V(e^{i \omega_2}) \right] e^{i \omega_1 m}e^{i \omega_2 n} d\omega_1 d\omega_2 = \\
&= \cep{u} (m) \: \delta(n) + \cep{v} (n) \: \delta(m) 
\end{aligned}
\end{equation}
\end{small}
We calculate then $1$D log cepstrum of $\hel{f}(p)$, the helical transform of $f$:
\begin{small}
\begin{equation}\nonumber
\begin{aligned}
&\log[\hel{F}(z)] = \log\left[\sum_{p=0}^{\infty} \hel{f}(p) z^{-p}\right]= \log[F(z, z^M)] = \\
&= \log\left[ \sum_{m=0}^{M-1}  u(m)z^{-m} \sum_{n=0}^{\infty} v(n)  z^{-Mn}\right] = \\
&= \log\left[ \sum_{m=0}^{M-1}  u(m)z^{-m}\right] + \log\left[ \sum_{n=0}^{\infty} v(n)  z^{-Mn}\right] = \\
&= \log[U(z)] + \log[V(z^M)]
\end{aligned}
\end{equation}
\end{small}
Thus the cepstrum of a separable function is given by
\begin{small}
\begin{equation}\nonumber
\begin{aligned}
\hel{\cep{f}}(p) &= \frac{1}{2 \pi i} \ointctrclockwise_{\lvert z \lvert =1} \log \left[\hel{F}(z) \right] z^{p-1} dz = \\
&= \frac{1}{2 \pi} \int_{-\pi}^{\pi} \log \left[\hel{F}(e^{i \omega}) \right] e^{i \omega p} d\omega = \\
&= \frac{1}{2 \pi} \int_{-\pi}^{\pi} \left\{ \log \left[U(e^{i \omega}) \right] + \log \left[V(e^{i \omega M}) \right] \right\} e^{i \omega p} d\omega = \\
&= \cep{u}(p) + \frac{1}{M} \: \cep{v}\left(\frac{p}{M}\right)
\end{aligned}
\end{equation}\end{small}
\end{proof}

\end{propo}
The same conclusions hold for a separable function of three or more variables. 
\begin{small}
\begin{figure}[h!]\centering
	\includegraphics[scale=0.3]{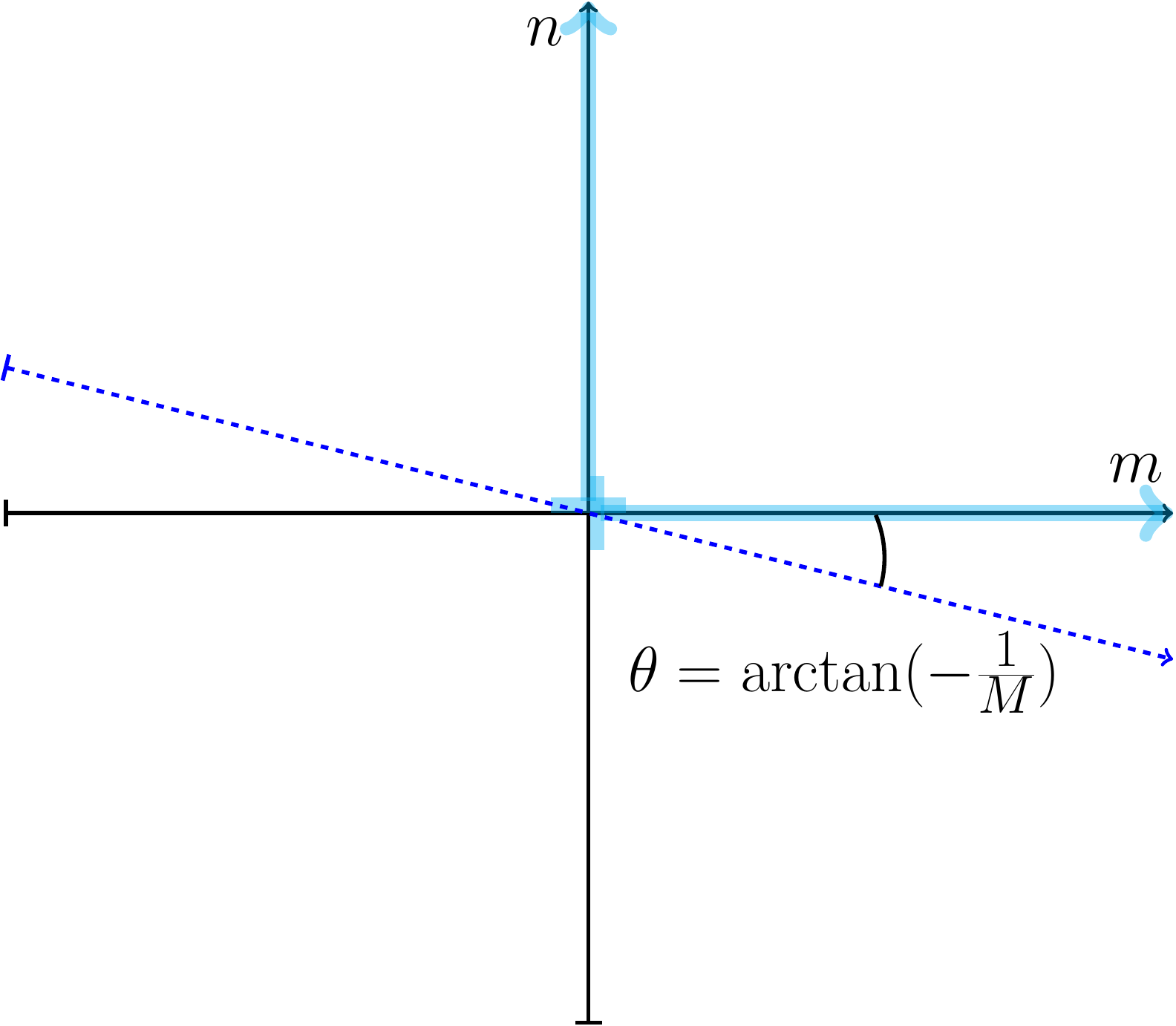}
	\caption{Semi-causal cepstrum for a separable function}\label{fig-inverse-helix-cepstrum-separable}
\end{figure}
\end{small}

\bigskip
We also give an alternative proof of these facts in the Z-domain, in the case of separable functions of three variables: the periodization along one dimension corresponds to a re-mapping and increase in number of poles and zeros of the Z-transform, that nevertheless maintain the same modulus. 

\begin{proof}[Proof]
Let us calculate the Z-transform of a finite sequence of three variables $f(n_x, n_y, n_t)$, with $0 \leq n_x \leq N_x$, $0 \leq n_y \leq N_y$ and $n_t \in \mathbb{N}$: 
\begin{small}
$$
F(z_x, z_y, z_t) = \sum_{n_x=0}^{N_x-1} \sum_{n_y=0}^{N_y-1} \sum_{n_t=0}^{\infty}
f(n_x,n_y,n_t) \cdot z_x^{-n_x} z_y^{-n_y} z_t^{-n_t}
$$
\end{small}
Helical boundary conditions are defined through the helical bijection 
$$
\begin{aligned}	
\phi: & [0:N_x-1] \times [0:N_y-1] \times \mathbb{N} \longrightarrow \mathbb{N} \\
& (n_x,n_y,n_t) \longmapsto n_z = N_x(N_y n_t + n_y) + n_x
\end{aligned}
$$
Starting from the original $3$D function, we can thus define the $1$D helical vectorization (or \textit{helix}) as
$
\hel{f}= f \circ \phi 
$,
with Z-transform 
\begin{small}
$$
\hel{F}(z) = \sum_{n_z = 0}^{\infty} \hel{f} (n_z) z ^{-n_z}
$$
\end{small}
We can then express $\hel{F}$ in relation to $F$ as
\begin{small}
\begin{equation} \nonumber
\begin{aligned}
\hel{F}(z) &= \sum_{n_x=0}^{N_x-1} \sum_{n_y=0}^{N_y-1} \sum_{n_t=0}^{\infty} 
\hel{f} \left(N_x(N_y n_t + n_y) + n_x \right) \cdot z^{-\left(N_x(N_y n_t + n_y) + n_x \right)} = \\
&= \sum_{n_x=0}^{N_x-1} \sum_{n_y=0}^{N_y-1} \sum_{n_t=0}^{\infty}
f (n_x,n_y,n_t) \: z^{-n_x} \left(z^{N_x}\right)^{-n_y} \left(z^{N_x N_y}\right)^{-n_t} 
\end{aligned}
\end{equation}
\end{small}
Therefore, 
\begin{equation} \label{eq-hel-Z-rel-3D}
\hel{F}(z)= F(z, z^{N_x}, z^{N_x N_y} )
\end{equation}
Let us consider the polynomial expression of $F$, in the case of a separable function $f$. For sake of simplicity, the polynomial function is assumed to be an all-zeros function, with a finite number of roots:
\begin{small}
\begin{equation}\nonumber
\displaystyle{F(z_x, z_y, z_t) = A \:
\prod_{i_x=1}^{N_{\alpha,x}} \left( z_x - \alpha_{i_x} \right)
\prod_{i_y=1}^{N_{\alpha,y}} \left( z_y - \alpha_{i_y} \right)
\prod_{i_t=1}^{N_{\alpha,t}} \left( z_t - \alpha_{i_t} \right)}
\end{equation}
\end{small}
On the other hand, from (\ref{eq-hel-Z-rel-3D}) we derive the polynomial expression of $\hel{F}$:
\begin{small}
\begin{equation}\nonumber
\displaystyle{\hel{F}(z) = A \:
\prod_{i_x=1}^{N_{\alpha,x}} \left( z - \alpha_{i_x} \right)
\prod_{i_y=1}^{N_{\alpha,y}} \left( z^{N_x} - \alpha_{i_y} \right)
\prod_{i_t=1}^{N_{\alpha,t}} \left( z^{N_x N_y} - \alpha_{i_t} \right)}
\end{equation}
\end{small}
The following remarks can be made:
\begin{enumerate}
\item 
$\hel{F}$ shares its zeros $\alpha_{i_x}$ with $F$.
\item 
For each zero $\alpha_{i_y}$ of $F$, $\hel{F}$ has $N_x$ corresponding new zeros, 
$
\hel{\alpha}_{i_y,k} = \left| \alpha_{i_y}  \right|^{1/N_x} e^{i 2 \pi k / N_x} \\
$
\item For each zero $\alpha_{i_t}$ of $F$, $\hel{F}$ has $N_x N_y$ corresponding new zeros, 
$
\hel{\alpha}_{i_y,l} = \left| \alpha_{i_y}  \right|^{1/(N_x N_y)} e^{i 2 \pi l / (N_x N_y)} \\
$
\end{enumerate}
Consequently, the zeros of $\hel{F}$ lie inside the unit circle if and only if the zeros of $F$ are inside the unit circle. The same considerations can be made for a rational function $F$ with poles at the denominator. 
\end{proof}

Thus, this result can be easily generalized to $d$D and leads to the following statement:
\begin{propo}
If the variables of $f(n_1, n_2, ..., n_d)$ separate, $1$D minimum phase of its helical transform $\hel{f}$ is equivalent to strict $d$D minimum phase of the $d$D sequence $f$. This equivalence is always verified, not only asymptotically. 
\end{propo}

\section{Helical mapping for wavefield propagation} \label{sec-back-prop}
PDEs describing wave propagation generally have two possible solutions: $f_+$ is forward propagating and then causal, $f_-$ is back propagating and then anti-causal; see Eq. (\ref{eq-wave-eq-visc-abs}) for an example.
Let $f(x,t)$ be the general solution of a PDE describing wave propagation, sampled in space $x_m=m \Delta$ and time $t_n = n T$. We want to recover the causal solution of the wave equation through spectral factorization with helical mapping, on the basis of the following result: 
\begin{propo}\label{prop-back-prop}
The helical processing of the data matrix $f(m,n)$ can lead to the cancellation of the back propagating solution of the PDE, if the helical vectorization $\hel{f}(p)$ is performed column-wise, \textit{i.e.} $p = m+Mn$ (thus periodizing $f(m,n)$ with respect to space). 
\end{propo}

\begin{proof}
Since the back-propagating term coincides with the forward (causal) propagating solution $f_{\oplus +}$ reversed in time, it represents the semi-maximum phase component of the power spectral density, $f_{\ominus -}$. 
Now, thank to Proposition \ref{prop-asymp-equiv}, if the helix is constructed through periodization with respect to space, the minimum-phase term $\hel{f}_+$ will asymptotically correspond (for large $M$ and after remapping to the $2$D space), to the semi-minimum phase term of the two factor decomposition $f_{\oplus +}$. Thus, according to Corollary \ref{corol-MP-time}, it will be minimum phase with respect to time and approximate the forward propagating solution. 
\end{proof}

In the frequency domain, the extraction of the semi-minimum phase solution $f_{\oplus +}$ is equivalent to applying an all-pass phase-only filter to the data. This is consistent with \cite{Gari79:ieeetsp}, where the boundary condition (at the surface) of the wave equation needs to cancel the back propagating solution through convolution with a semi-causal filter.

In order to illustrate Proposition \ref{prop-back-prop} with a straigthforward example, we make the assumption of an homogeneous medium, with constant propagation speed $c$. However, we want to integrate the viscosity $\alpha$ and absorbance  $\beta$ of the medium into the $1$D wave equation (this translates into attenuation of waves in space and time):
\begin{equation}\label{eq-wave-eq-visc-abs}
c^2 \frac{\partial^2 \Phi}{\partial x^2} = \frac{\partial^2 \Phi}{\partial t^2} + \alpha \frac{\partial \Phi}{\partial x} + \beta \frac{\partial \Phi}{\partial t}  
\end{equation}
The general solution of (\ref{eq-wave-eq-visc-abs}) is expressed by
$$
\Phi(x,t) = \mathcal{F}(x-ct) + \mathcal{G}(x+ct)
$$
For a plane wave equation corresponding to eigenmode $\omega$, this yields
$$
f_\omega(x,t) = A_0 e^{i (k x - \omega t)}e^{-\alpha x} e^{-\beta t} + B_0 e^{i (k x + \omega t)}e^{-\alpha x} e^{\beta t}
$$
The causal solution is embedded in the first term ($B_0=0$). After sampling with periods $\Delta$ for space and $T$ for time, the continuous and discrete causal solutions of (\ref{eq-wave-eq-visc-abs}) take the expressions 
\begin{equation}\label{eq-propag-wave-2D}
\begin{cases}
f(x,t) = A_0 e^{-\alpha x} e^{-\beta t}  e^{i k x} e^{-i \omega t}   \\
f(m,n) = A_0  e^{-\alpha m \Delta} e^{-\beta n T} e^{i k m \Delta} e^{-i \omega n T} 
\end{cases}
\end{equation}
The attenuated propagating wave in (\ref{eq-propag-wave-2D}) can be considered as the impulse response of the propagative system: $f(m,n) = \delta(m,n) \ast h(m,n) = h(m,n)$.
\ref{app-propagative-systems} details the computation of the poles and zeros of the Z-transform of the causal solution of the wave equation, and discusses the effects of the helical mapping in the Z domain. Furthermore, \ref{app-propagative-systems} shows how the back-propagating solution of (\ref{eq-wave-eq-visc-abs}) corresponds to the minimum phase term, reversed in time. 

\section{Application to physical systems} \label{sec-application}
Helical coordinates have been used in helioseismology \cite{RickC00:SP} for the estimation of a minimum phase impulse response. More generally, physical environments involving the propagation of waves, like the interior of the sun for helioseismology or the Earth volume for passive seismic, can be represented as convolutive systems \cite{TakaNFDSLRT12:ieeespm}.
Simulated data are generated by a convolutive propagative system 
$
d(x,t) = s(x,t) \ast h(x,t)
$
where $s(x,t)$ refers to the excitation signature and $h(x,t)$ to the impulse response. 
The FT of the data matrix is then expressed as the product
$
D(k_x,\omega) = S(k_x,\omega) \: H(k_x,\omega)
$.
In the examples presented in this paper,  we aim at estimating the acoustic impulse response of the Sun, $h(x,t)$, including internal reverberations. We make the assumptions that seismic excitations $s(x,t)$, generated by small  \textit{sunquakes}, are uncorrelated in space and time, so that the  power spectral density of $d(x,t)$, $\lvert S_d(\omega_x, \omega) \lvert^2$, equals $ \lvert H(\omega_x,\omega) \lvert^2$ up to a scale factor, and that $h$ is semi-minimum phase.  In \ref{app-algo} we detail the two algorithms used for comparisons: on one hand the $d$D spectral factorization (Algorithm \ref{algo-2D-sp-fact-procedure}), on the other hand the helical spectral factorization (Algorithm \ref{algo-helical-procedure}). 

Figure \ref{fig-2D-deconv} (a) and (b) show simulated data for $M=N=1024$ and the impulse response of the system, modeled as a Ricker wavelet \cite{GholK14:ieeescies}:
$$
h(x,t) \propto \frac{1}{\sqrt{2\pi}\sigma^2} \left\{1- \frac{[t-\tau(x)]^2}{\sigma^2}  \right\} e^{- \frac{[t-\tau(x)]^2}{2\sigma^2} },
$$
where $\tau(x) = \sqrt{x^2 + R^2}/v$,
and $R$ indicates the distance of the source. Temporal and spatial sampling periods are fixed at $20 ms$ and $5 m$, and $\sigma = 0.01$. The random excitation is modeled as a Gaussian white noise  with unit variance, in both dimensions: $s \sim \mathcal{N}(0,\mathbb{I})$.

Figures \ref{fig-error-Dirac} (a) and \ref{fig-error} (a) show the estimated impulse responses $\hat{h}$ through the helical procedure described in Algorithm \ref{algo-helical-procedure}; Figures \ref{fig-error-Dirac} (b) and \ref{fig-error} (b) show the distribution of the estimation error with respect to the true impulse response $h$. 

Figure \ref{fig-MSE-approx} (a) shows the total approximation error $e^{tot} = \lvert\lvert \hel{f}_+(p) - \hel{f_+(m,n)} \lvert\lvert^2$ of the helical minimum phase solution (Algorithm \ref{algo-helical-procedure}) with respect to the $2$D semi-minimum phase solution (Algorithm \ref{algo-2D-sp-fact-procedure}), as a decreasing function of the number of space samples $M$. This can be interpreted as a confirmation of the asymptotic equivalence of the two solutions stated in Prop. \ref{prop-asymp-equiv}.
Another measure of the quality of the approximation is expressed by the correlation coefficient $R$ between the two solutions, illustrated in Figure \ref{fig-MSE-approx} (b) as an increasing function of the number of space samples $M$.

The algorithm was then applied to $3$D solar data in Figure \ref{fig-solar-cube} (for more information about the experimental setup, cf. \cite{Solar_data}), and the Sun impulse response was estimated through helical spectral factorization (in Figure \ref{fig-solar-estimation} (a) and (b), we present our results for a given time instant and space location). Figure \ref{fig-solar-Pearson} shows the correlation coefficient between the helical and the $3$D solutions, as a function of the number of samples along the $x$ and $y$-axes, and along the time axis.
The estimated impulse response with multiple reflections is consistent with results in  \cite{RickC00:SP} and it seems to be related to a propagative system where seismic waves are reverberated at least three times within the Sun (cf. Figure \ref{fig-solar-estimation} (b)).

\section*{Conclusion}
This paper gives a theoretical foundation to the relevance of helical boundary conditions,  \textit{i.e.} a generalization of the vectorization of a multidimensional array, for the spectral factorization problem. Effects of this representation are detailed in the cepstral domain and in the Z domain, and the proposed technique is then illustrated through an example of blind deconvolution for a propagative system, and an application to helioseismology.

\section*{Aknowledgments}
This work has been partially supported by the ERC grant 2013-320594  ``DECODA''.

\section*{Figures}

\begin{figure}[h!]
\centering
	\subfloat[Data $d(x,t) = h(x,t) \ast s(x,t)$]{\includegraphics[scale=0.30]{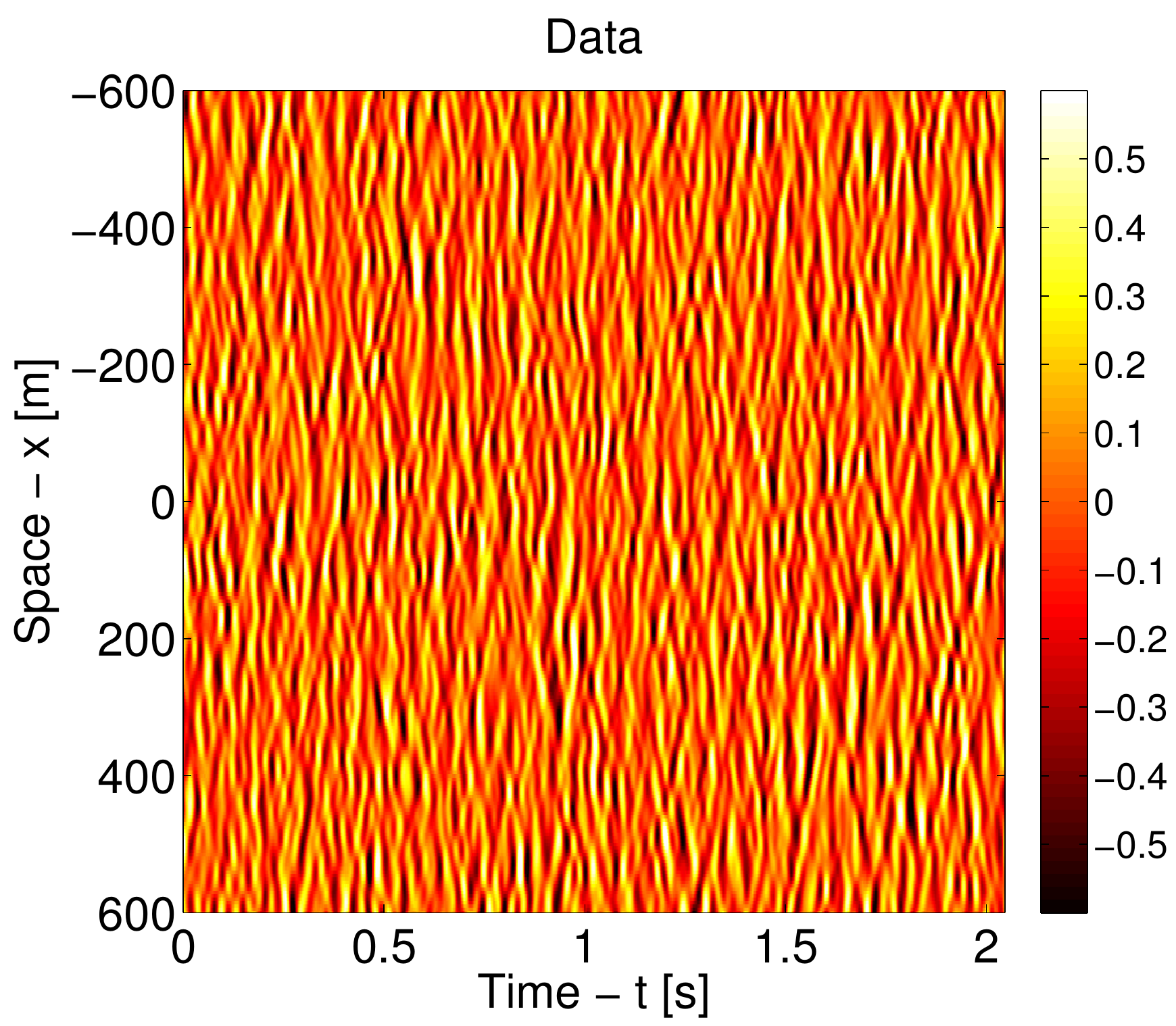}}
	\hspace{0.5em}
	\subfloat[Impulse response $h(x,t)$]{\includegraphics[scale=0.30]{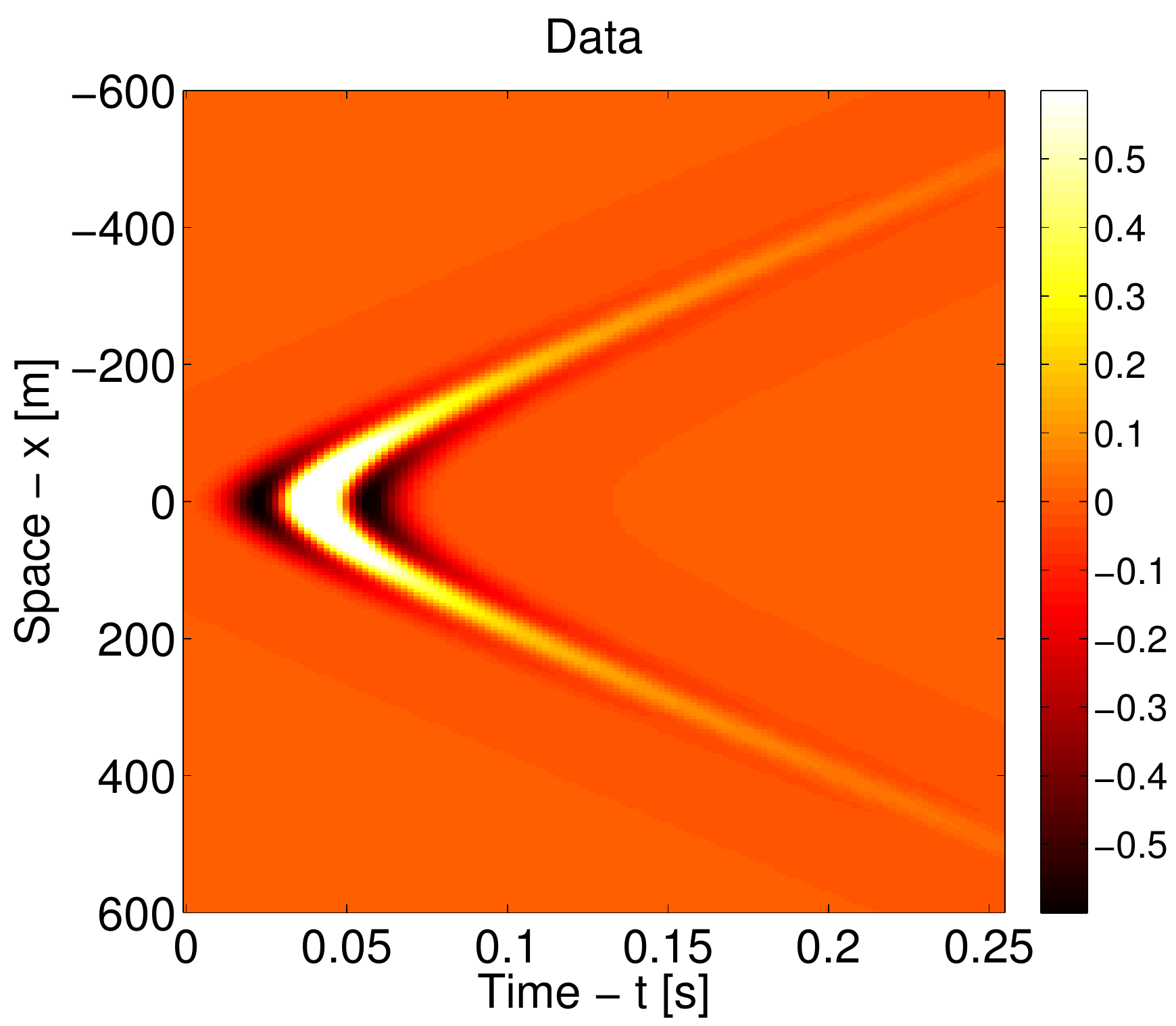}	\label{fig-quarter-R--}}
\caption{Simulated 2D data and impulse response in the plane $(x,t)$}\label{fig-2D-deconv}
\end{figure}

\begin{figure}[h!]
\centering
	\subfloat[$\hat{h}(x,t)$]{\includegraphics[scale=0.30]{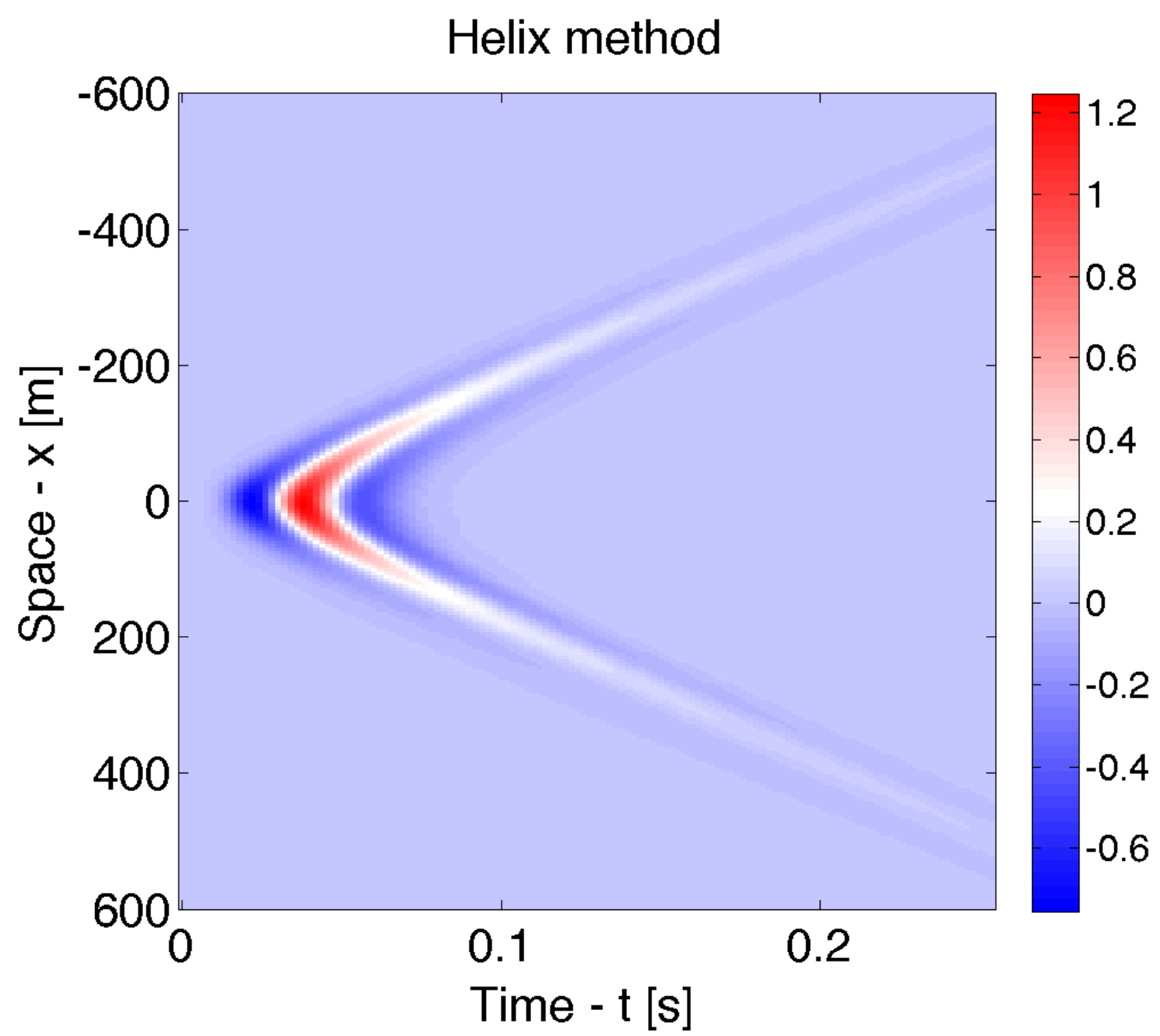}}
	\hspace{0.5em}
	\subfloat[Distribution of the error]
	{\includegraphics[scale=0.30]{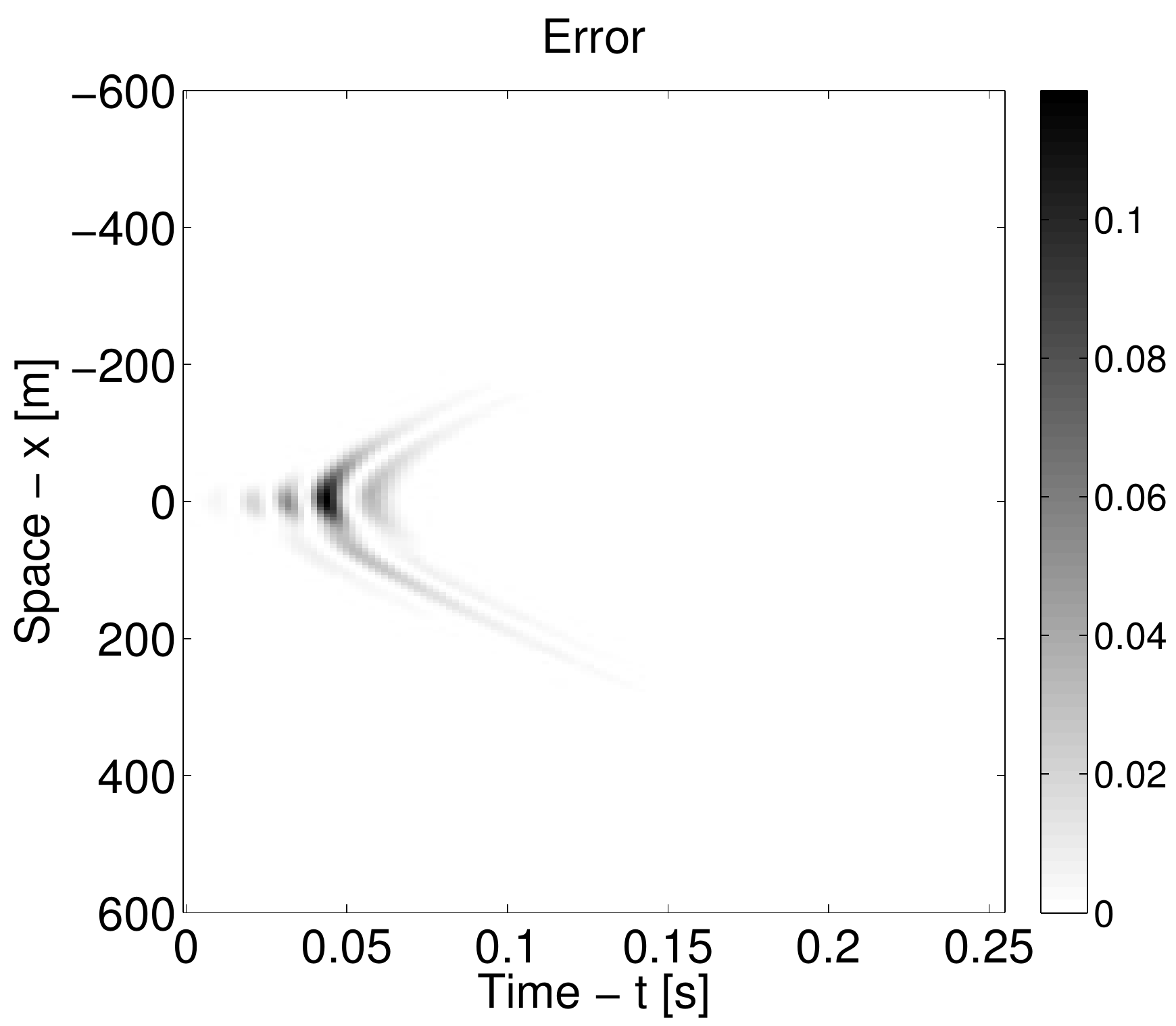}}
\caption{Estimation of the impulse response - only one Dirac source $\delta(x,t)$}\label{fig-error-Dirac}
\end{figure}

\begin{figure}[h!]
\centering
	\subfloat[$\hat{h}(x,t)$]{\includegraphics[scale=0.30]{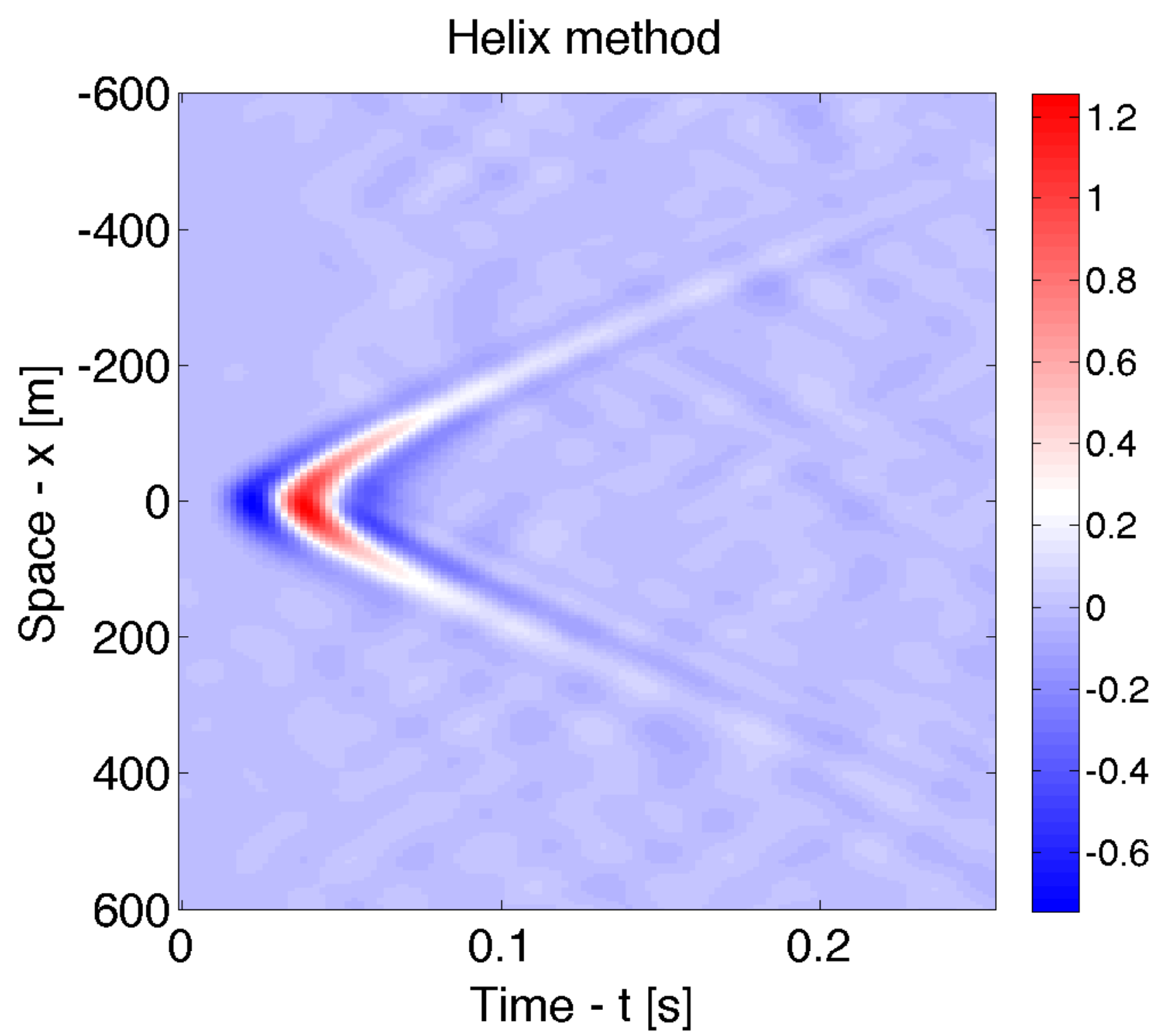}}
	\hspace{0.5em}
	\subfloat[Distribution of the error]
	{\includegraphics[scale=0.30]{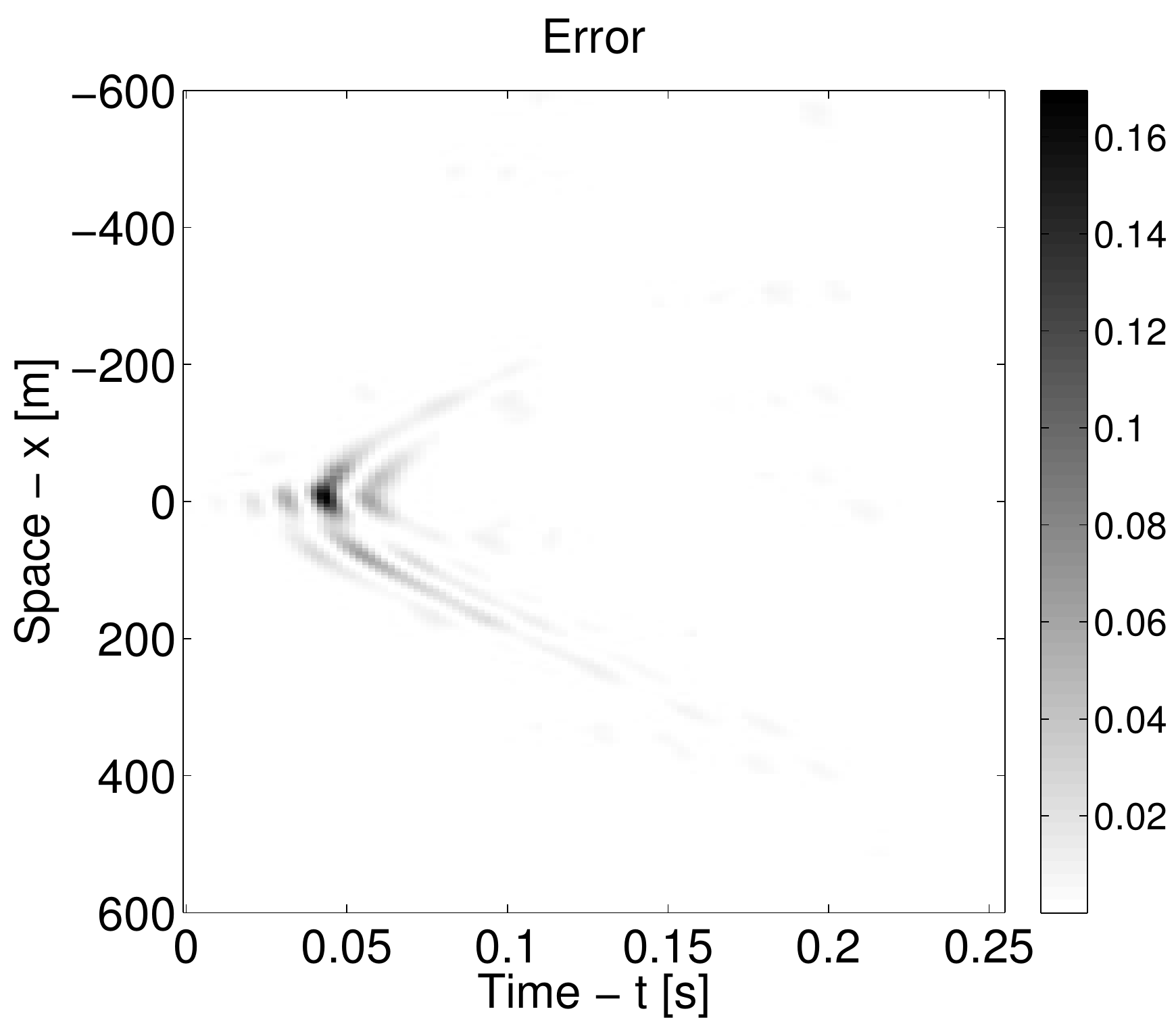}}
\caption{Estimation of the impulse response - random excitation $s(x,t)$}\label{fig-error}
\end{figure}

\begin{figure}[h!]
\centering
	\subfloat[Total approximation error vs $M$]{\includegraphics[scale=0.6]{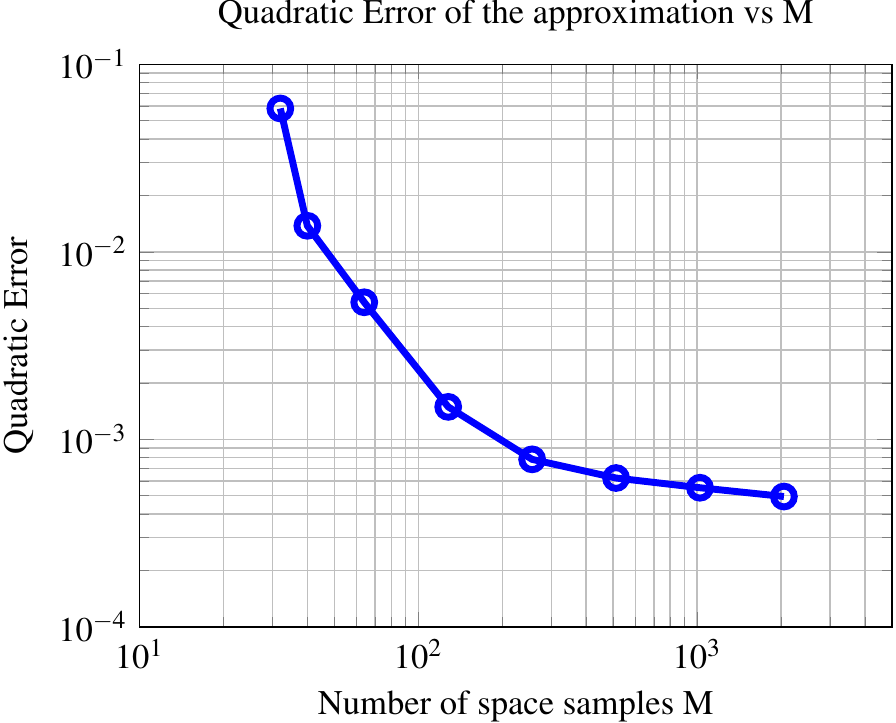}}\\
	%\hspace{2em}
	\subfloat[Pearson correlation vs M]{\includegraphics[scale=0.6]{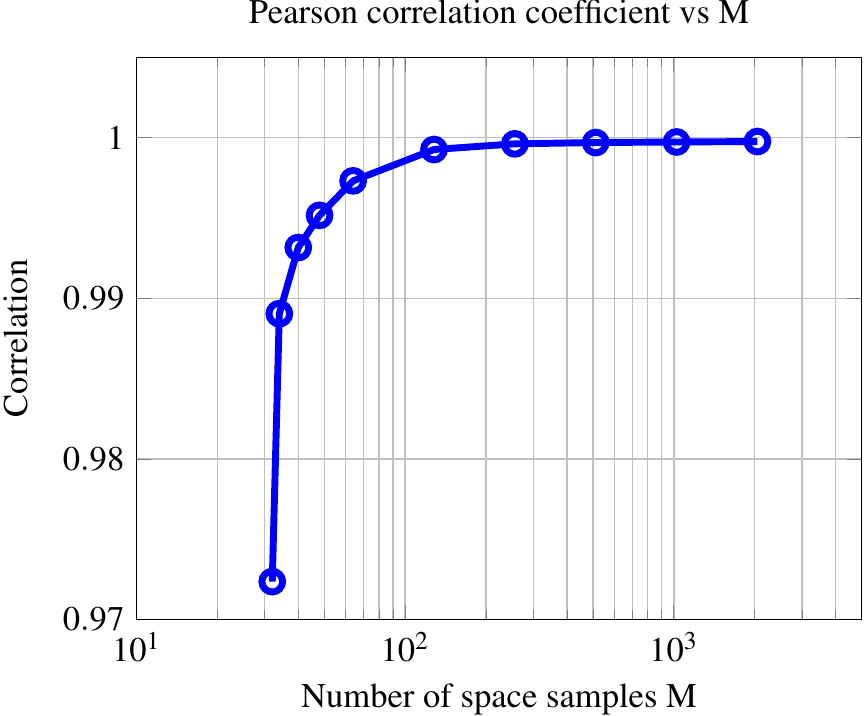}}
\caption{Approximation error and correlation with respect to the $2$D solution}\label{fig-MSE-approx}
\end{figure}

\begin{figure}[h!]
\centering
	\includegraphics[scale=0.37]{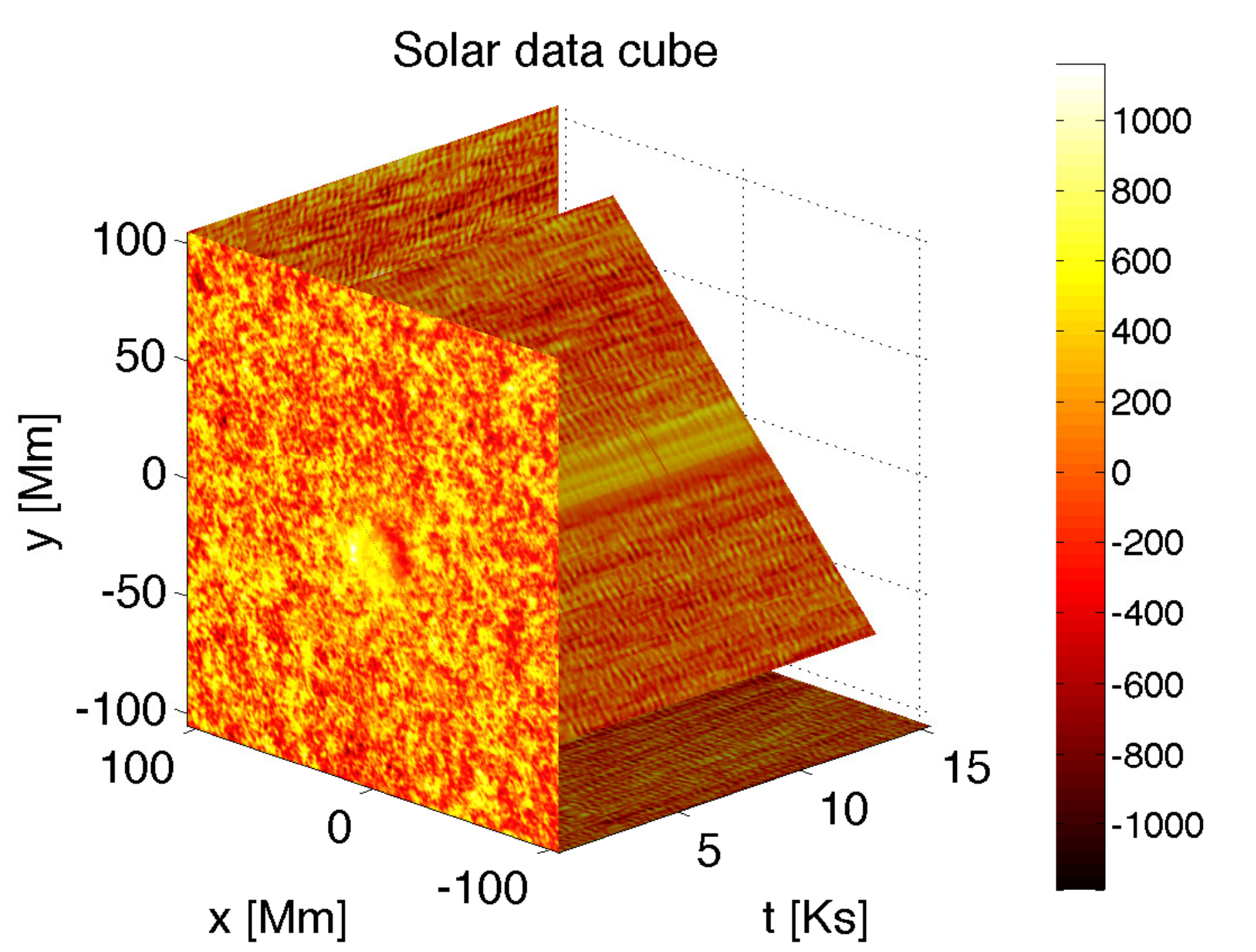}
	\caption{Solar Data Volume (Courtesy by Jon Claerbout, Stanford University)}\label{fig-solar-cube}
\end{figure}

\begin{figure}[h!]
\centering
	\subfloat[Impulse response for $t_0 = 1.8 ks$]
	{\includegraphics[scale=0.33]{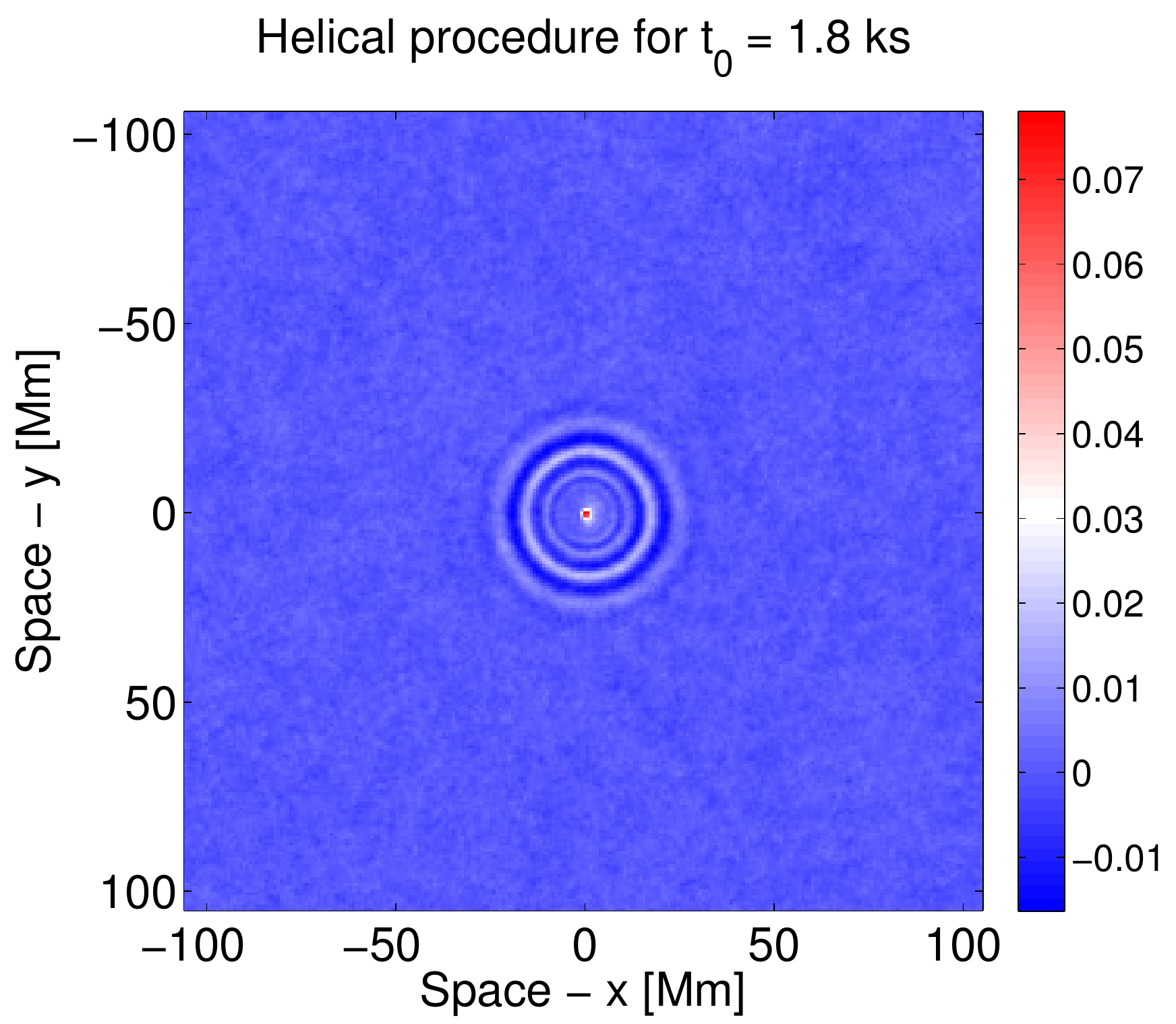}}\\
	%\hspace{0.05em}
	\subfloat[Impulse response for $y_0 = -3.3 Mm$]
	{\includegraphics[scale=0.33]{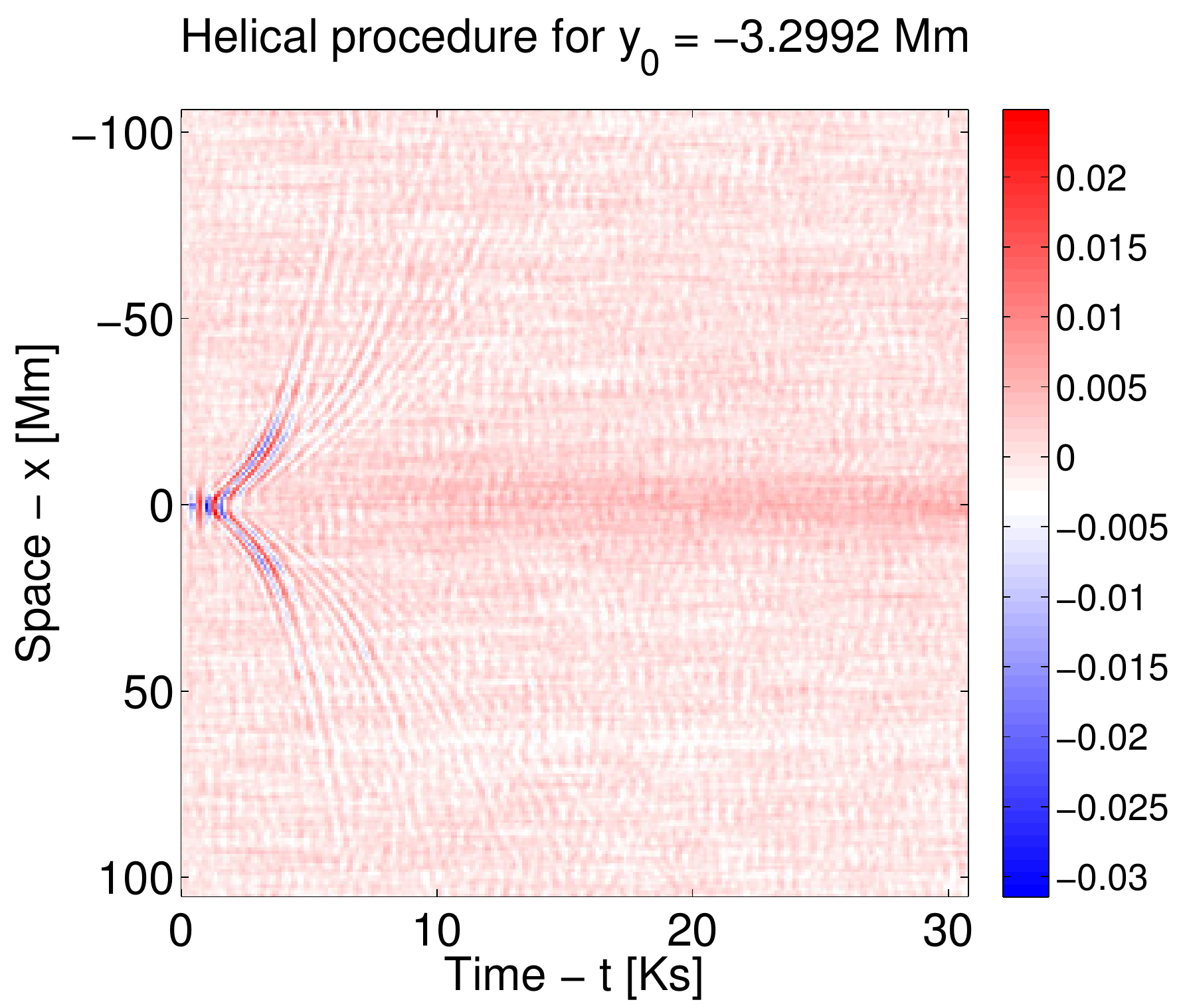}}
\caption{Estimation of the impulse response of the reverberations of the Sun.}\label{fig-solar-estimation}
\end{figure}

\begin{figure}[h!]
\centering
	\includegraphics[scale=0.55]{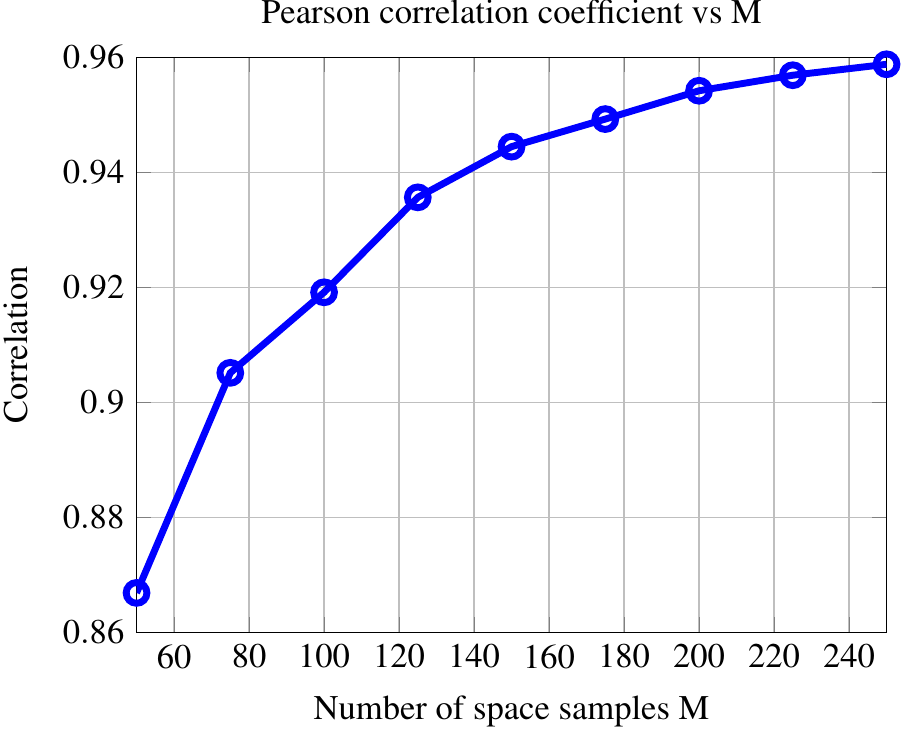}
\caption{Pearson correlation between the helical and the $3$D solutions vs M.}\label{fig-Sun}\label{fig-solar-Pearson}
\end{figure}

\clearpage
\appendix

\section{Proof for propagative systems}\label{app-propagative-systems}
We calculate the Z-transform of the forward propagating solution in (\ref{eq-propag-wave-2D}), $f(m,n) = A_0  e^{-\alpha m \Delta} e^{-\beta n T} e^{i k m \Delta} e^{-i \omega n T}$. If we have $m, n \in \mathbb{N}$, the result is simple:
\begin{equation}\nonumber
\begin{aligned}
&F(z_1, z_2) = \sum_{m=0}^{\infty} \sum_{n=0}^{\infty} f(m,n) z_1^{-m} z_2^{-n} = \\
&= A_0 \sum_{m=0}^{\infty} e^{-\alpha m \Delta} e^{i k m \Delta} z_1^{-m} \sum_{n=0}^{\infty} e^{-\beta n T} e^{-i \omega n T} z_2^{-n} \\
&= A_0 \frac{1}{1 - e^{-\alpha \Delta} e^{i k \Delta} z_1^{-1}} \frac{1}{1 - e^{-\beta T}e^{-i \omega T} z_2^{-1}}
\end{aligned}
\end{equation}
The convergence zone is given by
\begin{equation}\nonumber
\begin{cases}
\lvert e^{-\alpha \Delta} e^{i k \Delta} z_1^{-1} \lvert < 1  \Rightarrow  \lvert z_1 \lvert > e^{-\alpha \Delta} \\
\lvert e^{-\beta T} e^{-i \omega T} z_2^{-1} \lvert < 1  \Rightarrow  \lvert z_2 \lvert >  e^{-\beta T} 
\end{cases}
\end{equation}
The poles of $f(m,n)$ are given by
\begin{equation}
\begin{cases}
z_1 = e^{-\alpha \Delta} e^{i k \Delta}  \\
z_2 = e^{-\beta T}  e^{-i \omega T} 
\end{cases}
\end{equation}
If $0 \leq m \leq M-1$, numerator has supplementary zeros:
\begin{equation}\nonumber
{z_1}_k = e^{-\alpha \Delta} e^{i k \Delta + k \: 2 \pi /M}, \: 0 \leq k \leq M
\end{equation}
If $0 \leq n \leq N-1$, numerator has supplementary zeros:
\begin{equation}\nonumber
{z_2}_k = e^{-\beta T} e^{-i \omega T + k \: 2 \pi /N}, \: 0 \leq k \leq N
\end{equation}
Notice that $f(m,n)$ is strictly $2$D minimum phase, as its poles and zeros lie inside the unit bicircle. The back-propagating term would not be minimum phase with respect to time, due to poles $z^-_2 = e^{\beta T}  e^{i \omega T}$ with $\lvert z^-_2 \lvert >1$, and zeros ${z_2}^-_k = e^{\beta T } e^{i \omega T + k 2 \pi /N}$ with $\lvert {z_2}^-_k \lvert >1, \forall k$.

On the other hand, provided $0 \leq m \leq M-1$, after the helical mapping in (\ref{eq-helix-2}), the $1$D Z-transform of the helix $\hel{f}(p)$ is given by 
\begin{equation}\nonumber
\begin{aligned}
&\hel{F}(z) = \sum_{p=0}^{\infty} \hel{f}(p) z^{-p} = F(z, z^M) = \\
&= \sum_{m=0}^{M-1} \sum_{n=1}^{\infty} f(m,n) z^{-m} z^{-Mn} = \\
&= A_0 \sum_{m=0}^{M-1} e^{-\alpha m \Delta} e^{i k m \Delta} z^{-m} \sum_{n=1}^{\infty}  e^{-\beta n T} e^{-i \omega n T}  z^{-Mn}\\
&= A_0 \frac{1-e^{-\alpha \Delta M} e^{i k \Delta M} z^{-M}}{1-e^{-\alpha \Delta} e^{i k \Delta} z^{-1}}\frac{1}{1-e^{-\beta T}e^{-i \omega T} z^{-M}}
\end{aligned}
\end{equation}
The convergence zone for the Z-transform is $\{\lvert z \lvert >  e^{-\alpha \Delta},  \lvert z \lvert >  \sqrt[M]{e^{-\beta T}} \}$.
The poles are given by
\begin{equation}\nonumber
\begin{cases}
z = e^{-\alpha \Delta} e^{i k \Delta} \\
z_k = \sqrt[M]{e^{-\beta T}} e^{-i \omega T/M + k \: 2 \pi /M}, 0 \leq k \leq M
\end{cases}
\end{equation}
and the zeros by
\begin{equation}\nonumber
z_k = e^{-\alpha \Delta} e^{i k \Delta + k \: 2 \pi /M}, 0 \leq k \leq M
\end{equation}
If $0 \leq n \leq N-1$, numerator has supplementary zeros:
\begin{equation}\nonumber
z_k = \sqrt[M]{e^{-\beta T }} e^{-i \omega T / M+ k \: 2 \pi /(MN)}, 
0 \leq k \leq MN 
\end{equation}
Notice that $\hel{f}(p)$ is minimum phase, as its poles and zeros lie inside the unit circle. After helical mapping, the back-propagating term would not be minimum phase with respect to time, due to poles $z^-_k = \sqrt[M]{e^{\beta T}} e^{i \omega T/M + k \: 2 \pi /M}$ with $\lvert z^-_k \lvert >1$, and zeros $z^-_k = \sqrt[M]{e^{\beta T}} e^{i \omega T / M+ k \: 2 \pi /(MN)}$ with $\lvert z^-_k \lvert >1, \forall k$.

\section{Algorithms}\label{app-algo}
\begin{algorithm}[h]
\caption{$2$D spectral factorization}
\label{algo-2D-sp-fact-procedure}
\begin{algorithmic}[1]

\STATE Calculate $2$D spectrum $S(k,l), \: 0 \leq k \leq M-1, \: 0 \leq l \leq N-1$:

\begin{equation}\nonumber
S(k,l) = \left \lvert \sum_{m=0}^{M-1} \sum_{n=0}^{N-1} s(m,n) \: e^{-2 \pi m k /M} \: e^{- 2 \pi n l /N} \right \lvert^2
\end{equation}

\STATE Calculate $2$D complex cepstrum 

\begin{equation}\nonumber
\cep{s}(m,n) = \frac{1}{MN} \sum_{k=0}^{M-1} \sum_{l=0}^{N-1}\log S(k,l) \: e^{2 \pi m k /M}\: e^{2 \pi n l /N}
\end{equation}

\STATE Project the cepstrum onto the upper NSHP 
\newline
$
\mathcal{R}_{\oplus+} = \{m \geq 0, n \geq 0\} \cup \{m < 0, n > 0\}
$:
\begin{small}
\begin{equation}\nonumber
\cep{s}_+(m,n) = 0 \text{ for } (m,n) \in \mathcal{R}_{\ominus-} \coloneqq \{m \leq 0, n \leq 0\} \cup \{m > 0, n < 0\}
\end{equation}
\end{small}

\STATE Perform the inverse homomorphic transform on the semi-causal cepstrum to find the $2$D semi-minimum phase component:
$
s_{\oplus+}(m,n) = \mathcal{H}^{-1} [\cep{s}_{\oplus+}(m,n))]
$
\end{algorithmic}
\end{algorithm}

%\vspace{-5cm}
\begin{algorithm}[h]
\caption{Helical spectral factorization}
\label{algo-helical-procedure}
\begin{algorithmic}[1]
\STATE Vectorize data $\hel{s}(p)=f(m,n)$ column-wise, 
\newline
with $p = m+Mn, \: 0 \leq p \leq MN-1$

\STATE Calculate $1$D spectrum $\hel{S}(k), \: 0 \leq k \leq MN-1$:
$$
\hel{S}(k) = \left\lvert \sum_{p=0}^{MN-1} \hel{s}(p) \: e^{-2 \pi p k /(MN)} \right\lvert^2
$$

\STATE Calculate $1$D complex cepstrum 

\begin{equation}\nonumber
\hel{\cep{s}}(p) = \frac{1}{MN} \sum_{k=0}^{MN-1} \log S(k) \: e^{2 \pi p k /(MN)} 
\end{equation}

\STATE Project the cepstrum onto the $1$D causal admissible region $\{p \geq 0\}$:
$
\hel{\cep{s}}_+(p) = 0 \text{ for } p<0
$
\STATE Perform the inverse homomorphic transform on the causal cepstrum to find the $1$D minimum phase component:
$
\hel{s}_+(p) = \mathcal{H}^{-1} [\hel{\cep{s}}_+(p)]
$
\STATE Back project the helical minimum phase solution to the $2$D domain:
$
s^{\, helix}_+(m,n) = \hel{s}_+(m + Mn)
$
\end{algorithmic}
\end{algorithm}

\clearpage
\section*{References}

\bibliographystyle{elsarticle-num}
\bibliography{biblio}

%%%%%%%%%%%%%%%%%%%%%%%%%%%%%%%%%%%%%%%%%%%%%%%%%%%%%%%

%%%%%%%%%%%%%%%%%%%%%%%%%%%%%%%%%%%%%%%%%%%%%%%%%%%%%%%

%\IEEEpubidadjcol

%%%%%%%%%%%%%%%%%%%%%%%%%%%%%%%%%

%\subsection*{Acknowledgment}
%The authors would like to thank the reviewers for their insightful comments.

%%%%%%%%%%%%%%%%%%%%%%%%%%%%%%%%%%%%%%%%%%%%%%%%%%%%%%%%%%%%%%%%%%

%%%%%%%%%%%%%%%%%%%%%%%%%%%%%%%%%%%%%%%%%%%%%%%%%%%%%%%%%%%%%%%%%%%%%%%%%%%%%
%%%%%%%%%%%%%%%%%%%%%%%%%%%%%%%%%%%%%%%%%%%%%%%%%%%%%%%%%%%%%%%%%%%%%%%%%%%%%
%%%%%%%%%%%%%%%%%%%%%%%%%%%%%%%%%%%%%%%%%%%%%%%%%%%%%%%%%%%%%%%%%%%%%%%%%%%%%

\end{document}